\documentclass{amsart}		

\usepackage{amsfonts}
\usepackage{amssymb}

\newtheorem{theorem}{Theorem}[section]
\newtheorem{definition}[theorem]{Definition}
\newtheorem{example}[theorem]{Example}
\newtheorem{lemma}[theorem]{Lemma}
\newtheorem{proposition}[theorem]{Proposition}
\newtheorem{corollary}[theorem]{Corollary}
\newtheorem{remark}[theorem]{remark}
\newtheorem{problem}[theorem]{Problem}

\newcommand{\cala}[0]{{\mathcal A}}
\newcommand{\cals}[0]{{\mathcal S}}
\newcommand{\calr}[0]{{\mathcal R}}
\newcommand{\calx}[0]{{\mathcal X}}
\newcommand{\calv}[0]{{\mathcal V}}
\newcommand{\call}[0]{{\mathcal L}}
\newcommand{\calp}[0]{{\mathcal P}}
\newcommand{\calt}[0]{{\mathcal T}}

\newcommand{\N}[0]{{\mathbb N}}
\newcommand{\R}[0]{{\mathbb R}}
\newcommand{\Z}[0]{{\mathbb Z}}
\newcommand{\C}[0]{{\mathbb C}}

\newcommand{\ad}[0]{{\wedge}}
\newcommand{\ud}[0]{{\vee}}   
\newcommand{\defn}[0]{{:=}}
\newcommand{\ol}[0]{\overline}
\newcommand{\na}[0]{\nabla}
\newcommand{\e}[0]{\emptyset}
\newcommand{\tiha}[1]{{\tilde { \hat{ #1}}}}
\newcommand{\inft}[1]{{ {#1}^{(\infty)}}}

\newcommand{\re}[1]{{(\ref{#1})}}
\newcommand{\tihadot}[1]{{\tilde { \hat{ \dot{ #1}}}}}
\newcommand{\me}[1]{{\check{#1}}}
\newcommand{\tihabar}[1]{{\bar { \tilde{ \hat{ #1}}}}}
\newcommand{\ul}[0]{\underline}
\newcommand{\be}[1]{\begin{equation} \label{#1}}
\newcommand{\en}[0]{\end{equation}}

\author[B. Burgstaller]{Bernhard Burgstaller}
\email{bernhardburgstaller@yahoo.de}

\title{Information content in formal languages}


\address{Departamento de Matematica,  
Universidade Federal de Santa Catarina, CEP 88.040-900
Florianópolis-SC, Brasil }

\begin{document}




\begin{abstract}
Motivated by 
creating physical theories, 
formal languages $S$ with variables 
are considered 
and a kind of distance 
between elements of the languages is defined by
the formula $d(x,y)= \ell(x \na y) - \ell(x) \ad \ell(y)$,
where $\ell$ is a length function and $x \na y$ means the united theory of $x$ and $y$. 
Actually we mainly consider abstract abelian idempotent monoids 
$(S,\na)$ provided with length functions $\ell$. 
The set of length functions can be projected to another set of length functions 
such that the distance $d$ is actually a pseudometric and satisfies
$d(x\na a,y\na b) \le d(x,y) + d(a,b)$. 
We also propose a  ``signed measure'' on the set of Boolean expressions of elements in $S$, and a Banach-Mazur-like distance 
between abelian, idempotent monoids  with length functions, 
or formal languages.

\end{abstract}

\maketitle 

\tableofcontents

\section{Introduction}	

Variables are the center and core of 
mathematics, physics, 
theoretical technics  
and theoretical economy theory.  
In this paper we consider formal languages with variables 
(see definition \ref{deflangua}),
provided with a length function for sentences (for example a simple character count), and define something which appears like a
pseudo-metric between the sentences within one languages  
(see definition \ref{def210}).  

The languages are assumed to be equipped with an equivalence relation (def. \ref{defequ}), and the more 
two sentences $x,y$ have in common, the closer their distance becomes. 

Actually, this project was motivated   by physics in the quest
for a solution to quantum gravity, see for instance 
Weinberg \cite{weinberg2} 
and Hawking and Israel \cite{zbMATH03660360}. It is observed, that in general physical theories are surprisingly short when formulated as equations, 
and that progress in advancing the correct equations in physics is
often realized by small modifications in formulas. 
For example, the formulas between Hamiltonian mechanics and
quantum mechanics are very close in distance, also in our approach here, for all given Hamiltonians, even if quantum mechanics and mechanics appear very different when considered without 
formulas  (cf. example \ref{ex70}). 

We 
guess that a solution that unites gravity and quantum field theory, when formulated as mathematical equations, is actually close 
as formulas, for example in the sense proposed in this paper,
to the existing ununited theories of gravity and quantum field theory. 
This might 
lead to the direction towards a candidate for a united theory, and on the other hand, might be used as a tool to 
evaluate a given candidate theory by computing the 
distance to the existing theory so far. 
 
 In 
 short terms, the simpler a theory is, the better. The closer it is to existing effective theories,  
 the better, as the more likely it is the correct generalization.   
 
 In this sense, even quantum mechanics 
  might very theoretically have had 
 been predicted without any inspection of black body radiation  or double slit experiment, because it is perhaps the closest theory to mechanics which avoids the infinities. 
 Simply by looking at formulas. 
 
 In this sense also, 
 a solution to quantum gravity 
 is 
 { perhaps} less a physical problem but more a pure mathematical problem. 
 It may be formulated like so: indicate the 
 theory with the 
 shortest distance - in the sense of this paper, say - to the existing theory, such that no 
 infinities occur (cf. problem \ref{prob1}). 
Even more so, as ultrahigh-energy experiments appear currently to be impossible by 
 several reasons, for example the no-return from black holes obstacle.  
 
 After this motivating and ambitious 
 remarks, which may 
 also appear somewhat 
 self-evident, let us come to the hard facts 
 of this paper:
 
 Our definitions may be used much more widely, for example in computing, and it is widely applicable, as the approach is very
 flexible and 
 does not bother with details but with the core principle.
 World applications are viewed as formulas.  
 Actually, the distance $d(s,t)$ between two formulas $s,t$ is defined as
 $$d(s,t) = \ell(s \na t) - \ell(s) \ad \ell(t)$$
 which means in words: unite both theories $s$ and $t$ (called $s \na t$), optimize their length (called 
 $\ell$) and compute it (the more overlap they have the shorter it will become) and then subtract the length of the shorter theory of $s$ and $t$. 
  
  It looks natural 
  and 
  simple, but is also intriguing because of the minus sign. Even if it appears natural that the triangle inequality - also called  $\Delta$-inequality in this paper - 
  holds for $d$ (cf. lemma \ref{lem16}), it is difficult and enormous involving when trying to verify this.  		
  A key of this paper is to repair the definition 
  of $d$ by making the $\Delta$-inequality valid, see definition \ref{def71}. 
    See also theorem \ref{thm61} for other conditions which imply
    the $\Delta$-inequality. 
  
  As another example, it is difficult to get any easy criteria 
  for the 
  uniform continuity of a homomorphism between two languages,
  for example computer languages, but see lemma 
  \ref{lem121}.   

To not get lost in details or become too specific, in this paper we 
  mainly forget languages at all but 
  instead work with an abelian, idempotent monoid $(S,\na)$ provided 
  with a length function $\ell$ (see definition \ref{def211}). 
  
  This setting can be somewhat compared with a measured space
  $(M,\cup)$ 		
  with the operation $\cup$ of taking the union of sets and a measure $\ell$ on it. And indeed, we will 
  make parallels, and show how we can compute the ``signed measure'' $\zeta$ (i.e. 
  length) of a formal Boolean expression  in 
  the  monoid $S$,		
  see definition \ref{def66}. 
  The above 
  formula for  $d$ is easily seen to be indeed a metric for measured spaces $(M,\cup)$ 
  and $\na := \cup$, see lemma \ref{lemma11}. 
 
 In a main result, see theorem \ref{cor52}, we show how one can, starting with a length function $\ell$, construct a length function $\inft \ell$ such that the above formula for $d$ with $\ell$ replaced by $\inft \ell$ is indeed a pseudometric, 
 and the measure $\zeta(x \cap y)$ of the Boolean expression 
 of intersection of two sentences $x$ and $y$, 
 %
%
 is monotonically increasing in $x$ and $y$, as one may 
 naturally  expect it. 
 We have however no clue in how far $\inft \ell$ degenerates, it could even be completely zero (but don't believe that). 
 

 A key 
 role 
 in this paper also plays the $\na$-inequality
 $$d( x \na y,  a \na b ) \le d(x,a) + d(y,b)$$
 which, even if it holds heuristically, see 
 lemma \ref{lem41}, we also force by definition finally, 
 see definition  \ref{def72}.   
  
  For homomorphisms between languages we use the same concepts as for languages itself to define a metric, enriched by  the  
  submultiplicativity (cf. lemma \ref{lemmulti}). 
  
  As a distance between languages, or more generally, 
  abelian, idempotent monoids, one may use the proposed 
  Banach-Mazur-like distance of definition \ref{def827}. 

  In the final sections we give various examples which 
  should illustrate and motivate the theory of this paper,  
  presented in an extremely  superficially 
   and sloppy way.   
  
  
  Actually, 
  defining pseudometrics 
  on languages, and on a family of languages 
  goes back for a long time with a wide literature, and we only 
  cite samples, far from being in any way complete,    
  Han and Ko \cite{zbMATH06763319},  
  {Djordjevi{\'c}, 
  Ikodinovi{\'c} 
  and Stojanovi{\'c} } 
  \cite{zbMATH07274642},
   {Fulop  
   and Kephart 
   } \cite{zbMATH06819795},  
   {Ko, 
   Han  
   and Salomaa}   
   \cite{zbMATH06275883},  
   Imaoka \cite{zbMATH01793825},   
  Cseresnyes and Seiwert \cite{zbMATH07466697},   
  Ibarra, 
  McQuillan and Ravikumar \cite{zbMATH07067598}, 
  Bertoni, Goldwurm and 
  Sabadini \cite{zbMATH04033108},  
  Stratford
   \cite{zbMATH05658818},  
  Jalobeanu 
  \cite{zbMATH00569758},   
  {Parker,  
  K.B. Yancey 
  and M.P. Yancey}    
  \cite{zbMATH07204337}, 
  {Cui, Dang, Fischer and Ibarra} 
  \cite{zbMATH06244084}, 
  Pighizzini  
  \cite{zbMATH01853156},  
  and 
  in the context of physical theories,
  Giangiacomo \cite{zbMATH04197945}. 
  We will not compare these approaches with the one in this paper,        
  but we think there is a different emphasis on the latter. 

  Let us 
  remark, that 
  we less stress the exact computability 
  of the proposed metrics, but the idealized relations 
  they should satisfy.  
  Also, our definitions are designed such that we find them 
  particularly useful for physical theories.

  
\section{Metric}		
							\label{sec2}



Under a language $(S,A)$ we mean a set $A$ called alphabet, 
and a set $S \subseteq \{ f:\{1,\ldots,n\} \rightarrow A| \,n \ge 1\} 
\cup \{\e\}$ 
of finite text-strings with letters in $A$ 
called sentences or formulas or theories. Here, $\e$ means the
empty 
sentence (i.e. empty set).  
 In this section we are going to give successively several definitions, each of which
 becomes part of the considered setting after the definition. 

\begin{definition}[Language with variables] 	
\label{deflangua}

{\rm 
A {\em language with variables} means a language $(S,A)$, where 
$\e \in S$ and  
$A$ contains at least a sign $\ud$ and a distinguished  
infinite subset of so-called variables $X \subseteq A$,
divided further into two disjoint infinite subsets $X_0 \cup X_1 = X$ ($X_0 \cap X_1 = \e$) called {\em main variables} set $X_0$ and {\em auxiliary 
variable} set $X_1$,  
%
%
such that 
for all elements $s,t \in S$ and variable permutations 
$\phi \in P_S$ 
(to be explained below),
 $$s \ud t \in S 	$$	
 $$\phi(s) \in S$$
 



}
\end{definition}

Here, $s \ud t$ means the concatenated string 
of the string $s$ with the letter $\ud$ with the string $t$ 
in the order as indicated.  
The letter $\ud$ serves purely as a symbol for concanation of strings. 

\begin{definition}[Variables permutation]

{\rm 

Let $\psi:X \rightarrow X$
be a permutation (i.e. bijective function) on the variable set such that $\psi(X_i) = X_i$ ($i=1,2$), and let 
$\phi:S \rightarrow S$ 
denote the canonically bijective map induced
by $\psi$. 
That is, set $\phi(a_1 \ldots a_n):= \psi(a_1) \ldots \psi (a_n)$ 
for $a_i \in A$, where $\psi(a):=a$ if $a \in A \backslash X$. 

Write $P_S$ for the set of all such functions 
$\phi:S \rightarrow S$, called variable permutation or transformation.  


}
\end{definition}

\begin{definition}[Length function]

{\rm 

Under a length or cost function on a language $S$ with variables
we mean a positive function $L: S \rightarrow \R$ 
%
satisfying
$$L(\emptyset)  = 0$$
$$ L(x \ud y) = L(x) + L(y)$$
$$L(s) = L(\phi(s))$$
%

for all $x,y \in S$ and $\phi \in P_S$. 

} 
\end{definition}

In other words, the empty word has no length, the length of two concanated strings is the sum of of the length of the single strings, and the length
or cost of a string does not change if we just choose other variable
names.  

\begin{lemma}
The map $f:S \times S \rightarrow S$ 
$$f(s,t)= s \ud t$$
is an associative semigroup operation on $S$.
\end{lemma}

Two sentences are completely unrelated if they have no variables in common because they do not speak 
about the same variables: 

\begin{definition}[Independence]

{\rm 

Two elements $x,y \in S$ are called {
independent}, notated by $x \bot y$, 
if and only if $x$ and $y$ have no
variables in common.   

(In other words, 
a variable showing up in the string $x$ must not show up in the string $y$.)


}
\end{definition}

In languages like mathematical of physical theories, one can 
typically express formulas in formally different ways but which 
are equivalent in content (for example, formally different 
equations  may have the same solution). That is why we need an equivalence relation
on the set of sentences, and the following axioms 
are the minimum natural ones we postulate: 

\begin{definition}[Equivalence relation]			\label{defequ}	

\label{equrel}

{\rm 

Let be given an equivalence relation $\equiv$ on $S$  
such that
the following relations hold for all $a,b,s,
t,x 
\in S$ and
%
$\phi \in P_S$:

Neutral element:
\begin{equation} 			\label{eq28}
s \ud \emptyset \equiv \emptyset \ud s  \equiv s
\end{equation}

Commutativity:
\begin{equation}     	\label{eq27}
a \ud s \ud t  \ud b \equiv a \ud t \ud s \ud b
\end{equation}


$\ud$ preserves equivalent elements in case of independence: 
\begin{equation}     	\label{eq3}
s \equiv t 
, \; s  \bot x ,\; t 
\bot x   \quad   
\Rightarrow   \quad s \ud x \equiv t 
\ud x
\end{equation}

A superfluous copy can be removed:
%
%
\begin{equation} 			\label{eq8}
\phi(x) \bot x
\quad \Rightarrow \quad x \ud \phi(x)   \equiv x 
\end{equation}


Variable transformation does not change:
\begin{equation} 			\label{eq10}
x \equiv \phi(x)
\end{equation}


If $x$ contains  no main variables 
then require:
\begin{equation} 			\label{eq9}
x \equiv \emptyset 
\end{equation}  
}

\end{definition}

Let us revisit some aspects of the last definition. 
Relation (\ref{eq27}) means it is unimportant in which
order we list our sentences separated by $\ud$, like it is unimportant
in which order we list the single formulas in a longer theory, or the collection of subroutines in a longer computer program. 
 The connection between the single sentences is regulated solely by the variables they have in common. 
 That is why relation (\ref{eq3}) says that 
 the semigroup operation $\ud$ of concatenation does {\em not} 
 preserve the equivalence
 relation in general, but only if the 
 concatenated sentences are unrelated.
 Thus we also needed to incorporate $a$ and $b$ in commutativity (\ref{eq27}).  
 Relation (\ref{eq10}) tells us that it is unimportant which variables
 we choose (but only as a whole sentence).
Relation (\ref{eq8}) says if a theory 
is restated and fused with the
original theory with other variables then essentially nothing new is said 
 and the so new formed theory collapses to the original theory. 
%
%
The relation \re{eq9} is actually theoretically  
unimportant in the sense that it does not shine to the later abstract approach stated in definition \ref{def211} below through and may thus be 
omitted, but it is essential in the examples that we shall consider 
and are the only difference and the point  
what the distinction between main and auxiliary variables 
is all about. 
(See also example \ref{ex101}.) 

We use 
class brackets as in $[x]$ for classes in quotients, where $x$ is a representative.   
The limits of sequences and the compare 
of all real valued functions $f:Y \rightarrow  \R$ are always understood to be taken pointwise (i.e. $f \le g$ iff $f(x) \le g(x)$ for all $x,y \in Y$). 
 
The whole theory of this paper is mainly build only on the 
quotient set of sentences: 

\begin{definition}[Quotient]			\label{def26}

{\rm 
Set 
$\ol S := S/\equiv$ to be the set-theoretical quotient.


}
\end{definition}

On the quotient of sentences we naturally form an operation 
of concatenation by concatenating different sentences by choosing copies of them with disjoint 
variables chosen: 

\begin{definition}[Operation]				\label{def27}

{\rm 

Define 
$\na: \ol S \times \ol S \rightarrow \ol S$ by 
$$[x] \nabla [y]:= [x \ud \phi(y)]$$ 
for all $x,y \in 
S$, where
$\phi \in P_S$ (dependent on $x,y$) is chosen such that $x \bot \phi(y)$. 

}
\end{definition}

We will then mainly work with an abstract abelian monoid that
actually comes out so far: 

\begin{lemma}			\label{lem02}
$(\ol S,\na,[\e])$ is 
an abelian, idempotent monoid (i.e. $x \na x = x$ for all $x \in \ol S$).  
\end{lemma}

\begin{proof}
Indeed, by \re{eq10} and \re{eq3} the 
above definition of $\na$ 
is independent of the chosen $\phi \in P_S$. 
 By \re{eq28},
 $[\e]$ is a null-element,  
 and by \re{eq27} we get commutativity. 
  Idempotence follows from \re{eq8}. 
\end{proof}

On the quotient set of sentences we essentially define 
a length function $\ell$ by choosing the shortest possible length $L$ 
among the equivalent representatives (i.e. $\ell([x]) = \inf_{y\equiv x} L(y)$), but 
to ensure that the length of $s \na t$ is always bigger or equal than the length of $s$ for $s,t$ in the quotient $\ol S$, we correct this 
definition by the formula of definition \ref{def74}.  
Thus the following definition of $\ell$ is just a somewhat obscured reformulation 
of defining $\ell$ at first as in definition \ref{def91} 
and then 
correcting it 
by definition \ref{def74}:  

\begin{definition}[Length function on quotient]			\label{def29}

Set ($x \in S$)
$$\ell: \overline S \rightarrow \R$$
$$\ell([x]) = \inf \{ L(y) \in \R|\,  y ,a \in S,  \;x \bot a, \;y \equiv x \ud a \}$$




\end{definition}

We shall conveniently 
denote $\e:=[\e] \in \ol S$. 
We get very natural properties for the length function
which we later will postulate axiomatically for abelian monoids:

\begin{lemma}			\label{lem03}
$\ell$ is positive 
and 
for all $x,y \in \ol S$ we have: 
%
\begin{equation}		\label{eql1}
\ell(\emptyset)=0
\end{equation}

Subadditivity:
\begin{equation}		\label{eql2}
\ell(x \na y) \le \ell (x) + \ell(y)
\end{equation}

Monotone increasingness: 
\begin{equation}		\label{eql3}
\ell(x) \le \ell(x \na y) 
\end{equation}

\end{lemma}

\begin{proof}

Let $\varepsilon > 0$. 
Given $x_1,x_2 \in S$, 
by definition \ref{def29} choose $y_1,y_2,a_1,a_2 \in S$ such that $\ell([x_i]) 
\ge L(y_i)+\varepsilon$ where $x_i \ud a_i \equiv
y_i$ and $x_i \bot a_i$ ($i=1,2$).

Subadditivity: Choose any $\phi \in P_S$ such that 
$y_1 \bot \phi(y_2) $. Then 
$$\ell([x_1])+\ell([x_2])  +2 \varepsilon
\ge L( y_1 ) + L( \phi(y_2))
=  L( y_1 \ud \phi(y_2))  
\ge \ell([y_1 ) \ud \phi(y_2)])$$
$$= \ell ([y_1]  \na  [y_2])  
=\ell ([x_1 \ud a_1]  \na  [x_2 \ud a_2])
=\ell ([x_1] \na  [x_2 ] \na [a_1]  \na [a_2]) 
\ge \ell([x_1] \na  [x_2 ])$$

where the last inequality follows from monotone increasingness,
which we prove next:
$$\ell([x] \na [a]) + \varepsilon =\ell([x \ud \phi(a)]) + \varepsilon
\ge L(y) \ge \ell([x])$$
where, given $x,a \in S$, $y,b \in S$ is chosen such that
$y \equiv x \ud \phi(a) \ud b$ 
and the before last inequality holds. 
%
\end{proof}

We often abbreviate the notion ``monotonically increasing'' 
for \re{eql3} simply by ``monotone'' or ``increasing''. 
Let $x \ad  y :=\min(x,y)$ and $x \ud y :=\max(x,y)$ denote the minimum and maximum, respectively, of two real numbers $x,y \in \R$. 

 The length function induces now two very natural
 notions of distance functions between sentences: 
 
\begin{definition}[``Metric'']			
				\label{def210}

{\rm

We define 
two 
functions 
$d,\sigma : \ol S \times  \ol S 
\rightarrow \R$ by  
$$d(s,t) = 
\ell(s   \na t) - \ell(s) \ad \ell(t) $$
$$\sigma(s,t)=  2 \ell(s \na t) - \ell(s)  - \ell(t)$$

}
\end{definition}

We sometimes call these functions also more precisely $d_\ell:=d$
and $\sigma_\ell:=\sigma$.   
They have the characteristic of 
distance 
between elements of $\ol S$: 

As an illustrative example
let us recall that
in his famous work \cite{zbMATH02579754}, Dirac took the then known 
quantum numbers valued Klein-Gordon 
function 
$\psi: \R^4 \rightarrow B(H)$ (operators on Hilbert space) defined on spacetime $\R^4$
and the Klein-Gordon differential equation $\Delta \psi - \partial_t \psi - m^2 \psi= 0 $ ($m \in \R$ constant)
and 
took the formal square-root of the Klein-Gorden differential 
 operator $K:=\Delta - \partial_t - m^2$ and formed the Dirac differential operator 
 $D:=\sqrt{K}$ with constant coefficients in matrix algebra $M_4(\C)$, that is,
 $D^2 = K$, and proposed the physical equation $D \psi = 0$
 and along with it predicted correctly the existence
 of particles and anti-particles 
 \cite{zbMATH03002233}. 
 This is a prototype example where an existing physical theory
 is slightly modified to obtain a  new theory. That is, the Klein-Gordon
 equation is close to the Dirac equation, also in our ``metric''. Let us compute this: 
 
\begin{example}				\label{ex213}

{\rm

The prototype example is physical theories. 
$S$ are the collections of all sentences, which are equations, formulas, function definitions and so on. Equivalence relation $\equiv$ is when two formulas are equivalent in a mathematical sense. The $\ud$ means ``and'' in the sense 
that this and the other formula/equation holds. 
The equivalence relation $\equiv$ has to be enriched by the axioms 
of definition \ref{defequ}.   
Main variables are those that cannot be deleted, like 
the electromagnetic field in Maxwell's equations, because it is what the theory is about. 
The length function $L$ could be a simple character count. 
See also the beginning of section \ref{sec10} including 
example \ref{ex101} for further explanations. 
  
  For example, consider the free Klein-Gordon equation $\Delta \psi =  \partial_t^2 \psi + m^2 \psi$ and the Dirac equation $D \psi =0$, where
  $D^2 = \Delta - \partial_t^2 - m^2 =:K$ , and where in both cases
  $\psi$ is a main variable. 
  %
  Then
  $$d([\mbox{Dirac}], [\mbox{Klein-Gordon}]) =
    d([D \psi = 0 \ud D^2 = \Delta - \partial_t^2 - m^2],
  [(\Delta - \partial_t^2- m^2) \psi = 0]) $$
    $$
    = \ell([D \psi = 0 \ud D^2 = \Delta - \partial_t^2 -m^2 
  \ud (\Delta - \partial_t^2 -m^2) \phi = 0]) 
  - \ell([(\Delta - \partial_t^2 -m^2) \psi = 0])
  $$
  $$= \ell([D \psi = 0 \ud D^2 = K \ud K:= \Delta - \partial_t^2 
  -m^2 
  \ud K \phi = 0]) 
  - \ell([(\Delta - \partial_t^2  -m^2) \psi = 0])$$
  $$\approx 23 - 13 = 10$$ 
  whereas $\ell([\mbox{Dirac}]) = 15$  
  and $\ell([\mbox{Klein-Gordon}]) \approx 13$. 

  Here, $K$ and $D$ are auxiliary variables.
  If we had a longer formula of $K$ (i.e. 
  the Klein-Gordon operator) then the above distance $d(D,KG)$ would not change, but $\ell(D)$ and $\ell(KG)$ would increase. 

Note that we did not explicitly transform the variables $x_1,x_2,x_3,t$ in $(\Delta - \partial_t^2) \phi(x,t)$ 	
as required by the $\na$-operation of definition \ref{def27} because they are 
bound variables which we can label 
 as we like without change, and so have back transformed. 
  
}  
  \end{example}	

\begin{example}						\label{ex22}
{\rm
Another 
standard example is computer languages.   
For example $C$++.  
Sentences are meaningful programs or code snippets.  
Two programs are equivalent, if they do the same,
for example, independent of time performance. 
Typically, the entry point of a program is a main variable. 
Set $L(p)$ 
to be the length of a program $p$. 
 Set $\ud$ 
 to be a concatenation of programs. 

Variables are the typical variables in higher computer languages. 
The variable transformation 
in the  $\na$-operation of definition \ref{def27}
ensures that the variables of two different programs do not overlap. 

The more two programs $s,t \in \ol S$ have in common, the lower 
$d(s,t)$ becomes. 

For example, if $s$ and $t$ have one common routine, then it is enough to keep one copy of them in $s \ud t$. 	

}
\end{example}




Throughout, if nothing else is said, we 
switch to the following general setting for $\ol S, \na$ and $\ell$. 
Only if we discuss something involving $S$ itself, 
we implicitly assume without saying that we use the above 
specific definitions of $\ol S,\na$ and $\ell$. 

\begin{definition}[More abstract definition]				\label{def211}

{\rm 

Let $(\ol S,\na,\e)$ be an abelian, idempotent monoid. 


We call $\ell : \ol S \rightarrow \R$ a length function 
if the assertions of lemma \ref{lem03} hold for it. 


}
\end{definition}

$\ell(x)$ is also called the 
{\em information content} of $x \in \ol S$. 

Positivity of $\ell$ is also automatic by idempotence and subadditivity, as $\ell(x) = \ell(x \na x) \le  2 \ell(x)$.  
%
Instead of the above specific metric candidates $d$ and $\sigma$
from above, we shall sometimes use the following generalization: 

\begin{definition}[Metric candidate]			 \label{def212}
We call 
$d : \ol S \times \ol S \rightarrow \R$ a positive, symmetric, 
nilpotent 
function if 
$$d(x,y) \ge 0 , \quad d(x,y) = d(y,x) , \quad d(x,x)=0$$

for all $x,y \in \ol S$. 

\end{definition}

$d$ is called a pseudometric on $\ol S$ if it is 
positive, symmetric, nilpotent and satisfies 
the $\Delta$-inequality (see definition \ref{def31}).

\section{$\Delta$-inequality}

The following section is a first discussion when 
the proposed distance functions $d$ and $\sigma$ of defintion 
\ref{def210} could indeed satisfy the well known triangle inequality 
(here also called $\Delta$-inequality). 
This is actually a main theme of this paper and first
satisfying answers are given in theorems \ref{thm61} 
and \ref{cor52}.

Let $(\ol S,\na,\e)$ be an  
abelian, idempotent monoid.


\begin{definition}		\label{def31}

{\rm 

A function 
$d: \ol S \times \ol S \rightarrow \R$ is said to satisfy the $\Delta$-inequality if
$$d( x ,z ) \le d(x,y) + d(y,z)$$
for all $x,y,z\in \ol S$ 
}
\end{definition}

\begin{lemma}

If a symmetric function $d: \ol S \times \ol S \rightarrow \R $ satisfies the $\Delta$-inequality then
also the second $\Delta$-inequality  	
(i.e. second triangle inequality)
$$|d(x,z) - d(z,y)| \le d(x,y)$$

\end{lemma}

To heuristically see 
that $d$ of definition \ref{def210} 
in the context of formal languages 
 could satisfy the $\Delta$-inequality 
we 
do the following considerations. 
Assume a sentence $s \in S$ is written in its shortest form
as $s\equiv s_1 \ud \ldots \ud s_n$ for $s_i \in S$, that is, such that $\ell([s]) \approx 
L(s_1 \ud \ldots \ud s_n) =L(s_1)+ \ldots L_(s_n)$. 
Similarly we do so for a given sentence $t\equiv t_1 \ud \ldots t_m$. 
  Many $s_i$ could be auxiliary function definitions, some of which may be identical to some $t_j$. 
  As a first approximation we may estimate
  $$\ell([s] \na [t])= \ell([s \ud \phi(t)])
  = \ell([s_1 \ud \ldots s_n \ud \phi(t_1 ) 
  \ud \ldots \phi(t_m)])$$
 $$ =\ell([s_1 \ud \ldots s_n\ud \psi(t_1 ) 
  \ud \ldots \psi(t_k)]) 
  \le L(s_1) + \ldots + L(s_n) + L(t_1 ) + \ldots +L(t_k)   
  $$
  
  where in the last line we have dropped identical auxiliary function definitions 
  $\phi(t_{k+1}) , \ldots , \phi(t_m)$ which appear already in the $s_i$s, and so instead of calling $\phi(t_j)$ we may call $s_i$ and that is why we 
  had to introduce another transformation $\psi$ and dropped the superfluous $\phi(t_j)$ by \re{eq9}.  
  In a certain sense, we formed the set theoretical union of
  $\{s_1,\ldots,s_n\}$ and $\{t_1,\ldots,t_m\}$. 
  
  If we sloppily replace the above $\le$ by $\approx$ and view $L$
  as a measure, then the $\Delta$-inequality holds:
  
  

\begin{lemma}			\label{lemma11}
Let $(X,\cala,\mu)$ be a measured space. 
Set $\cala_0 := \{A \in \cala|\, \mu(A) < \infty\}$ 
to  be the collection of finite, measurable sets. 
 %
Then 
\begin{eqnarray*}
d(A,B)   &:=&  \mu( A \cup B ) - \mu(A)  \wedge \mu(B) 
\end{eqnarray*}
$$\sigma(A,B):= 2 \mu( A \cup B ) - \mu(A)  -\mu(B)$$
where $A,B \in \cala_0$, 
define metrics on $\cala_0$.

\end{lemma}

\begin{proof}

Generally we have 
\be{eq37}
\mu(X \cup Z) - \mu(X) 
= \mu(Z \backslash X)		= \mu(Z )-\mu(Z \cap X)
\en

We only prove the $d$-case. 
(i)
Assume $\mu(Y) \le \mu(X) \le \mu(Z)$. 
Then we have equivalences  
$$\mu(X \cup Z) - \mu(X) \le \mu(X \cup Y)  + \mu (Y \cup Z) 
- 2\mu(Y) $$
$$\mu(Z )-\mu(Z \cap X)  \le \mu(X)- \mu(X \cap Y)  + 
\mu(Z) - \mu (Z \cap Y)  $$
\begin{equation} 		\label{eq3b}
\mu(X \cap Y)  + 
\mu (Z \cap Y)  \le \mu(Z \cap X)  +  \mu(X) 
\end{equation}
But this is true because in general we have
$$\mu(X \cap Y)  + 
\mu (Z \cap Y) \le \mu(Y) + \mu (X\cap Y \cap Z)$$

(ii)
Assume $\mu(X) \le \mu(Y) \le \mu(Z)$.   
Doing the same as above leads to \re{eq3b} but with 
$\mu(X)$ replaced with $\mu(Y)$ and this is then also 
true by the above argument.


(iii)
Assuming $\mu(X) \le \mu(Z) \le \mu(Y)$ 
leads to \re{eq3b} 
but with 
$\mu(X)$ replaced with $2 \mu(Y) - \mu(X)$ and this is then also 
true. 
%
%
\end{proof}

Actually $d$ is like the maximum's norm or $\ell^\infty$-norm,
and $\sigma$ like the $\ell^1$-norm in $\R^n$ or space of real-valued sequences: 

\begin{lemma}				\label{lem12}

Continuing the setting of the last lemma, we have $d\le \sigma$ and 
$$d(A,B) = \max( \mu(A \backslash B),  \mu(B\backslash A))$$
$$\sigma(A,B)= \mu(A \backslash B) +   \mu (B\backslash A)$$
\end{lemma}

This is obvious by \re{eq37}. 
Actually, $\sigma$ as just formulated is 
well-known to be a metric, $d$ 
might be less used. 
We continue by imposing definitions \ref{def211} and \ref{def210}.


Switch back to the notions of definitions \ref{def211}
and \ref{def212} for $\ol S, \ell, d$ and $\sigma$. 

\begin{lemma}				\label{lem34}

For all $x,y \in \ol S$ we have 
%
$$\ell(x) \ud \ell(y) \le \ell( x \na y)  $$
%
%
%
$$|\ell(x)-\ell(y)| \le d(x,y) \le  \ell(x) \ud \ell(y)$$
$$|\ell(x)-\ell(y)| \le \sigma(x,y) \le  \ell(x) + \ell(y)$$
\begin{equation}			\label{eq7}
d(x,y) \le \sigma(x,y) \le 2 d(x,y)
\end{equation}
$$d(x,\emptyset)= \sigma(x,\emptyset)=\ell(x)$$
%




\end{lemma}

\begin{proof}

The first relation follows from \re{eql3}. 
If $\ell(x) \le \ell(y)$ then $|\ell(x) -\ell(y)| = \ell(y) -\ell(x) \le 
\ell(x \na  y) - \ell(x) = d(x,y) \le \ell(x) + \ell(y) -\ell(x)=\ell(y)$ 
by \re{eql3}, showing the second line.
The third line follows from the similarly proved  \re{eq7} and \re{eql3}.
\end{proof}

This criterion for the triangle inequality will be 
extensively used in section \ref{sec6}:

\begin{lemma}			\label{lem16}

$\sigma$ satisfies   
the $\Delta$-inequality if and only if 
%
$$\ell(x \na z) + \ell(y) \le \ell(x \na y) + \ell(y \na z)$$

 
 for all $x,y,z \in \ol S$. 
 
 Moreover, this equivalence is true for any given function $\ell :\ol S \rightarrow \R$. 
 
 
\end{lemma} 

\begin{proof}
Divide the $\Delta$-inequality spelled-out by $2$ to see the equivalence:  
$$2 \ell(x \na z) - \ell(x) - \ell(z)  
\le 2 \ell(x \na y) - \ell( x) -\ell(y) +  
2\ell (y \na z) - \ell(y) - \ell(z)$$
  
\end{proof}

%

Interestingly,  if $\sigma$ satisfies the triangle inequality then $\ell$
must be a length function: 
  
\begin{lemma}

Given any function $\ell: \ol S \rightarrow \R$, if $\sigma$
satisfies the $\Delta$-inequality, then $\ell$ is monotonically increasing.

Moreover it is  positive and 
subadditive if $\ell(\e) \ge 0$, and even a 
length function  (def. \ref{def211}) if $\ell(\e)=0$.

\end{lemma}

\begin{proof}

Set $z=\e$ in lemma \ref{lem16} to get monotonically increasingness,
$y=\e$ to get subaddtivity, and positivity by $0 \le \ell(\e) \le \ell(x)$. 
%
\end{proof}

\section{$\na$-inequality}		

An 
important ingredient of our theory is the $\na$-inequality 
defined and first discussed in this section. 
It is fundamental in the proof of lemma \ref{lem511}, and thus
for the whole section \ref{sec7}, theorem \ref{lem511} and
corollary \ref{cor724}.

Let $(\ol S,\na,\e)$ be 
an abelian, idempotent monoid.

 
\begin{definition}				\label{def41}

{\rm 
A function 
$d: \ol S \times \ol S \rightarrow \R$ is said to satisfy the $\na$-inequality if
$$d( x \na y,  a \na b ) \le d(x,a) + d(y,b)$$
for all $x,y,a,b \in \ol S$

} 
\end{definition}

This may be heuristically compared with normed vector spaces by setting
$d(x,y):= \|x-y\|$, $x \na y:= x+y$ and having
$$\|x+y - a -b\|\le \|x-a\| + \|y-b\|$$

Going back to the remark before lemma \ref{lemma11} and the setting of this lemma, we heuristically argue that 
the $\na$-equality may hold:  

\begin{lemma}					\label{lem41}

Let $d$ be as in lemma \ref{lemma11}. 
Then 
$$d( A \cup B,  C \cup D ) \le d(A,C) + d(B,D)$$
\end{lemma}

\begin{proof}

Write $|A|:=\mu(A)$. 
We have to show that 
\be{eq38}
|A \cup B \cup C \cup D| - |A \cup B| \ad | C \cup D|
\le |A \cup C |  -  |A| \ad | C| 
  + |B \cup D| - |B| \ad | D|
\en


By symmetry, just exchange $A \leftrightarrow C$ and
$B \leftrightarrow D$, it is enough to consider 
the case $|A \cup B| \le | C \cup D|$. 


Recalling \re{eq37}, 
the left hand side of \re{eq38} is 
$$|(C \cup D) \backslash (A \cup B)| = |(C \cup D )\cap  A^c \cap   B^c| $$
$$\le |C \cap  A^c | + |D \cap   B^c|
= | C \backslash A| + |D \backslash B|  $$

Case 
$|A| \le |C| , |B| \le |D|$: 
In this case, the last inequality verifies \re{eq38} by \re{eq37}.


Case 
$|A| \le |C| , |B| \ge |D|$: 
In this case, $|D \backslash B| \le  |B \backslash D|$. 
Hence the last inequality verifies \re{eq38}. 
%
%
%
The third case is analogous. 
\end{proof}

\begin{lemma}			\label{lem42}

If 
a nilpotent function $d$ satisfies the $\na$-inequality then ($\forall x,y,z \in \ol S$)  
\be{eq70}
d(x \na z, y \na z) \le d(x,y)
\en

\end{lemma}

We say that a function $d$ satisfies the 
{weak $\na$-inequality}
if it satisfies the inequality \re{eq70},  
and only the very weak one, if this holds only for $x \le y$
(see definition \ref{deforder}). 

We finally remark that given a (possibly infinite) sequence of pseudometrics in our setting, we may sum 
them up 
to get possibly a more faithful 
pseudometric:  
\begin{lemma}

The set of length functions on $\ol S$ 
is a positive cone, i.e. 
if $\ell_1,\ell_2$ are length functions, then also $\lambda_1 \ell_1 + \lambda_2 \ell_2$ for constants $\lambda_1,\lambda_2 \ge 0$ in $\R$. Similarly, the functions $d$ satisfying the $\na$-inequality 
(or the $\Delta$-inequality, or those of the form $\sigma= \sigma_\ell$) 
are a positive cone, 
and one has
$\sigma_{\lambda_1 \ell_1 + \lambda_2 \ell_2} = 
\lambda_1 \sigma_{\ell_1} + \lambda_2 \sigma_{\ell_2}$. 

\end{lemma}

\section{Order relation}

In this section we consider a natural order relation
on an abelian idempotent monoid induced by the 
monoid operation. 
This is not only important in various occasions in the theory
like in the monotone increasingness of $\ell$ and in various
aspects in the next section, but also a fundamental relation 
for information itself. It says when a 
theory is information-theoretically contained in another theory, and because we work with equivalence relations on theories, this works 
also for formally very different looking representatives of theories. 
This effect will even become more intense when we shall divide out equivalences induced by a pseudometric in definition \ref{def33}. 
    
In this section, if nothing else is said, we continue 
with the setting and notions 
 of definitions \ref{def211} and 
\ref{def210}.


\begin{definition}			\label{deforder}
Define an order relation $\le$ on $\ol S$ by 
$(\forall x,y \in \ol S)$
$$x \le y  \quad \Leftrightarrow  \quad x  \na 
y = y	$$

\end{definition}

The following three lemmas are straightforward to check,
and the reader should recall the monotone increasingness of $\ell$, 
see \re{eql3}, 
which is often used, as it is in lemmas \ref{lem54} and
\ref{lem35}.

\begin{lemma}

{\rm (i)} This is an order relation on $\ol S$. 

{ \rm (ii) }
For all $a,b,x,y \in \ol S$ we have 
$$a \le x , b \le y  \quad    \Rightarrow 
\quad a \nabla b \le x \nabla y$$


{\rm (iii)  }
$\ol S$ is a directed set as
%
$$\e \le a,b \le a \nabla b$$

\end{lemma}

The length function $\ell$ preserves the order and
$d$ and $\sigma$ have particular simple formulas for ordered
pairs:

\begin{lemma}				\label{lemma53}

For all $a,b,c \in \ol S$ we have 
$$a\le b  \quad \Rightarrow \quad \ell(a)\le \ell(b)$$
$$a\le b \quad \Rightarrow  \quad d(a,b)= \sigma(a,b) = \ell(b)-\ell(a)  $$
$$a \le b \le c \quad \Rightarrow \quad  d(a,c)  =  d(a,b)+ d(b,c),
\quad  \sigma(a,c)  = \sigma(a,b)+ \sigma(b,c)$$
\end{lemma}
 
 Here are some further relations ($\delta_y$ is defined below)
 for $d$ and $\sigma$ under ordered pairs:
 
 \begin{lemma}								\label{lem53}
 For all $a,b  ,x,y \in \ol S$ we have 
 \be{eq20}
 \sigma(x,x \nabla y) =  d(x,x \nabla y)  = - \delta_y(x)
 \en
%
 $$- \sigma(x,x \nabla y) + \sigma(y,x \nabla y) 
= \ell(x) - \ell(y)$$
%
%
 	$$
 	\sigma(a,b)= d(a , a\nabla b) + d(b , a\nabla b) 
 	= \sigma(a,a \nabla b) +\sigma(b,a \nabla b)
 	$$
 	%
$$
d(a,b)= d(a , a\nabla b)  \ud d(b , a\nabla b)   
= \sigma(a , a\nabla b)  \ud \sigma(b , a\nabla b)  
$$
 
 In other words, if we know the values of $d$ or $\sigma$
 for all $a \le b$, then we know them for all $a,b$,
 and for both $d$ and $\sigma$. 
 \end{lemma}

 
 The following lemma reverses lemma \ref{lemma53}. 
  It also states in words that a theory is contained 
 in another  theory  (i.e. $x \le y$) if and only if the 
 information content of the united theory 
 is the 
 information content of the bigger theory  
 (i.e. $\ell(x \na y) = \ell(y)$).

 
\begin{lemma}			\label{lem54}
For all $x,y \in \ol S$, the following assertions are equivalent:
$$d(x,y)= \ell(y) - \ell(x) 
\quad \Leftrightarrow \quad 
d(y, x \na y)=0
\quad \Leftrightarrow  \quad   \ell(x \na y) = \ell(y)$$
$$\Leftrightarrow \quad \sigma(x,y)= \ell(y) - \ell(x) 
\quad \Leftrightarrow \quad 
\sigma(y, x \na y)=0$$

If $d$ or $\sigma$ is faithful (i.e. for all $x,y \in \ol S$ one has $d(x,y)= 0 \Rightarrow x=y$) then 
these relations are also equivalent to 
$$x \le y$$

\end{lemma}

\begin{proof}
(i)
if $d(x,y)= \ell(y)-\ell(x) \ge 0$ then  
$\ell(y) \ge \ell(x)$ and so
$ d(x,y)= \ell(x \na y) -\ell(x)$  and thus
$\ell(x \na y) = \ell(y) \ge \ell(x)$ and so 
%
$$d(y,x\na y)= \ell(x \na y) - \ell(y)=0$$
These implications can be reversed. 

The same equivalence 
with $d$ replaced by  $\sigma$ is analogously proved. 

(ii)
If $d$ is faithful this is equivalent to saying that $y= x \na y$
%
%
%
\end{proof}

\begin{lemma}					\label{lem35}

For all $x,y \in \ol S$ we have the equivalences
$$d(x,y)= 0 \quad \Leftrightarrow \quad \sigma(x,y)=0  
\quad \Leftrightarrow \quad 
\ell(x \na y) = \ell(x)= \ell(y)$$

In particular, $d$ is faithful iff $\sigma$ is faithful. 
\end{lemma}

\begin{proof}
If $\sigma(x,y)= 2 \ell(x \na y)-\ell(x) -\ell(y) = 0 \le 0$ then,
because 
$\ell(x)+\ell(y) \le 2 \ell(x \na y)$, we have
$2 \ell(x) \le 2\ell(x \na y) = \ell(x) + \ell(y)$ and so
$\ell(x \na y)= \ell(x) = \ell(y)$. 
The first 
equivalence is by \re{eq7}. 
\end{proof}


Typically, later in this paper we will construct pseudometrics
satisfying the $\na$-inequality and in the next definition
and proposition we state how the given structures on the abelian idempotent monoid carry over 
naturally to the quotient when turning a pseudometric to 
a metric.

\begin{definition}[From pseudometric to metric]				\label{def33}

{\rm 
Assume that a pseudometric $d$ on $\ol S$ 
satisfies the $\na$-inequality. 		


Define an equivalence relation $\cong$ 
on $\ol S$ by 
%
$$x \cong y  \quad \Leftrightarrow   \quad  d(x,y) = 0 \qquad \forall x,y \in \ol S$$
Set $\me S := \ol S/\cong$ to be the set-theoretical quotient. 
Define operations, relations and functions on $\me S$ by   
(for all $x,y \in \ol S$) 
$$[x] \na [y] := [x \na y]$$
$$[x ] \le [y]     
: \Leftrightarrow [y] = [x] \na [y] $$
$$\me d([x ] , [y])  :=  d(x  , y ) $$
$$\me \ell([x ] )  := 
\check d([x ] , [\emptyset])  $$ 


}
\end{definition}

Then the so defined language $\me S$ with equivalence
relation satisfies the concepts of section \ref{sec2}
again:

\begin{proposition}					\label{lem36}

If a pseudo-metric $d$ on $\ol S$ 
satisfies the $\na$-inequality then: 			


{\rm (i)} 
The $\na$ operation passes to $\check S$, 
that is, the above 
$(\check S,\na ,[\emptyset])$ is an abelian, idempotent monoid with
order relation as in definition \ref{deforder},  
and is provided with a metric $\check d$ satisfying the $\na$-inequality. 


{\rm (ii)} 
In case $\ol S = S / 
\equiv$ as of  
definition \ref{def26},  
then 
we may canonically identify 
$\check S = S / \equiv_2$ for an equivalence
relation $\equiv_2$ on $S$ satisfying the same rules of definition \ref{defequ}. 
(So $\check S$ is again of the form $\ol S$.) 
  
{\rm   (iii)}  
  In case $d= d_\ell$ (or $d= \sigma_\ell$), the above defined
  $\check \ell$ is a length function on $\check S$ and 
  $$\me{d_\ell} = d_{\me \ell} 
  \qquad (\mbox{or } \quad \me{\sigma_\ell} = \sigma_{\me \ell}
  )$$

\end{proposition}

\begin{proof}
(i) 
The $\na$ operation is well-defined on $\tilde S$ as
if $x \cong x'$ and $y \cong y'$ then 
$$d(x \na y, x' \na y') \le d(x, x')  + d(y,  y')  = 0$$
 
 (ii) Let $\alpha: S \rightarrow \ol S / \cong$ the canonical map 
 and $\equiv_2$ the associated equivalence relation $x \equiv_2 y$
 if and only if $\alpha(x)=\alpha(y)$ for all $x,y \in S$. 
 As all rules of definition hold for $\equiv$, they hold also for $\equiv_2$, excepting (\ref{eq3}) needs explanation. 
 But  by lemma \ref{lem42} we have 
 $$s \equiv_2 t 
 \Leftrightarrow d([s],[  
 t ])= 0 
 \Rightarrow 0= d([s] \na [x],[  
 t] \na  [x])  
 = d([s \ud x],[ 
 t \ud x])	
 \Rightarrow s \ud x \equiv_2 
 t \ud x$$
for $s,  
t \bot x \in S$ 

(iii)
By definition 
$$\check \ell([x] \na [y]) - \check \ell([x]) = d_\ell(x \na y,\e ) - d_\ell(x,\e) 
= \ell(x \na y) - \ell(y) = d(x \na y,x) \ge 0$$
so $\check \ell$ is monotonically increasing, and its subadditivity 
follows from the $\na$-inequality of $\check d$. 
\end{proof}

\section{Measure of Boolean expressions}			\label{sec6} 

We have seen in the considerations before lemma \ref{lemma11}, that somehow, 
$\ell(x \na y)$ is the measure of the union of $x$ and $y$. 
In this section we want to extend the notion of a measure also 
to other 
Boolean expressions,  
like the intersection $x \cap y$ of sentences $x$ and $y$. 

There is not necessarily a sentence that represents the intersection or the set-theoretical difference $x \backslash y$ of two sentences
$x,y$, see example \ref{ex108}. 
But we can still compute the ``measure'' of these non-existing,
imagined elements.

Let us write formally $\zeta(E)$ for the measure of 
a 
Boolean expression $E$,
for example $\zeta(x \cap y)$. 
%
We will give precise definitions in this section. 

In this section, 
we assume the setting and notations of definitions \ref{def211} and 
\ref{def210}, 
but suppose that $\ell$ is possibly non-increasing, that is, 
we 
cease requiring \re{eql3}. 


\begin{definition}[Measure of intersection]		\label{def61}

{\rm 

For all $x,y \in \ol S$ we set 
$$
\zeta(x \cap y):= \ell(x) + \ell(y) - \ell(x \nabla y)$$

}
\end{definition}

An illustrative example for $\zeta(x \cap y)$ would be 
the amount 
two computer programs $x$ and $y$ have in common, for example
by identical or mathematically similar routines.
The following function slightly extracts some core of the last definition:

\begin{definition}			
[$\delta$-function]				\label{def62}

{\rm 

We also set, for all 
$x,y \in \ol S$,
$$\delta_y(x):=\ell(x) - \ell(x \nabla y)$$

}
\end{definition}

We say 
the { measure of intersection is increasing} if 
$$a \le x,b \le y \quad \Rightarrow  \quad \zeta(a \cap b) \le \zeta(x \cap y)$$
for all $a,b,x,y \in \ol S$, 
and { $\delta$ is increasing} if
$$\forall a,x,y \in \ol S: \qquad a \le x \quad \Rightarrow \quad \delta_y(a) 
\le \delta_y(x)$$

The following couple of lemmas will cummulate 
to theorem \ref{thm61} below.

\begin{lemma}					\label{lem61}
We have the following equivalences:   


The measure of intersection is increasing. 

$\Leftrightarrow$ $\zeta(x \cap y)$ is  increasing in $x \in \ol S$
for all fixed $y \in \ol S$.  

$\Leftrightarrow$ $\delta$ is  increasing.

\end{lemma}

\begin{proof}

If $\delta$ is increasing, then $\zeta(x \cap y)= \delta_y(x) + \ell(y)$
is increasing in $x$ for fixed $y$, and so also in $y$ for fixed $x$ by symmetry of $\zeta(x \cap y)$. Combining both we get that the measure of intersection is increasing. 
%
\end{proof}

\begin{lemma}					\label{lem62}

If 
$\delta$ is increasing and $\ell$ monotone,  
then $\sigma$ 
satisfies the $\Delta$-inequality. 



\end{lemma}

\begin{proof}

$\delta$ is increasing means that for all $x,y,z \in \ol S$ we have 
\begin{equation} \label{e2}
\ell(y ) - \ell(y \na x) \le \ell(y \na z) - \ell(y \na z\na x)  
\end{equation}
Monotony of $\ell$ implies 
$$\ell(z\na x) \le \ell(y \na z\na x)$$
Thus 
we get the criterion of lemma \ref{lem16}, that is, 
\begin{equation} \label{e1}
\ell( z\na x) + \ell(y )  \le \ell(y \na z) 
+ \ell(y \na x)
\end{equation}


\end{proof}

\begin{lemma}			\label{lem63}

If $\sigma$ satisfies the $\Delta$-inequality 
then $\delta$ is increasing.

\end{lemma}

\begin{proof}

Replace $z$ in \re{e1}, which holds by lemma \ref{lem16},  by $z \na y$ to obtain \re{e2}.
\end{proof}

\begin{lemma}					\label{lem64}

If 
$\delta$ is increasing 
then $\sigma$ satisfies the $\na$-inequality. 


\end{lemma}

\begin{proof}

By lemma \ref{lem61}, $\rho(x,y):= \zeta(x \cap y)$ is increasing in $x,y$ simultaneously.
Then  
$\rho$ is increasing implies 
$-\rho$ is decreasing, and applying this two times yields
$$-2\rho(x \na y,a \na b)= 2 \ell(x \na y \na a \na b)  - 2\ell(x \na y) - 2\ell(a \na b)$$
$$\le -\rho(x,a) -\rho(y,b) 
= \ell(x \na a) - \ell(x) - \ell(a)
+\ell(y \na b) - \ell(y) - \ell(b)$$

A simple rearrangement of the summands in this inequality yields
exactly 
$$\sigma(x \na a, y \na b) \le \sigma(x,y) + \sigma(a,b)$$

\end{proof}

\begin{lemma}			\label{lem65}

If $d$ or $\sigma$ satisfies the 
very weak $\na$-inequality and $\ell$ is monotone,
then
$\delta$ is increasing. 

\end{lemma}

\begin{proof}

The very weak $\na$-inequality of $d$ implies
$$- d(y ,y \na z) \le - d(y \na x, y \na z \na x)$$
which by the monotonicity of $\ell$ means 
$$\ell(y ) - \ell(y \na z) \le \ell(y \na x) - \ell(y \na z\na x)$$

But 
this means $\delta$ is increasing. 
The same 
 is true for $\sigma$ by \re{eq20}. 
%
\end{proof}


\begin{lemma}			\label{lem66}
If $\ell$ is monotone, then $d$ satisfies the very weak $\na$-inequality
if and only if $\sigma$ satisfies it.

\end{lemma}

\begin{proof}

Because of  \re{eq20}. 
\end{proof}

\begin{lemma}				\label{lem67}

If
$\sigma$ satisfies the $\Delta$-inequality   then 
$d$ satisfies the $\Delta$-inequality. 


\end{lemma}

\begin{proof}

Let $x,y,z \in \ol S$. 
In all cases like $\ell(x) \le \ell(y) \le \ell(z)$ etc. 
we see that
$$\ell(x) \ad \ell(y) + \ell(y) \ad \ell(z) 
- \ell(x) \ad \ell(z) \le \ell(y)$$ 

Thus, by lemma \ref{lem16} we get
%
$$\ell(x \na z) + \ell(x) \ad \ell(y) + \ell(y) \ad \ell(z) 
- \ell(x) \ad \ell(z) \le \ell(x \na y) + \ell(y \na z)$$

But this is the $\Delta$-inequality for $d$. 
\end{proof}

It is 
astonishing 
how strongly the validity of the $\Delta$-inequality 
and 
the $\na$-inequality of $d$ and $\sigma$, and the monotone increasingness of the measure of intersection and 
the 
$\delta$-function are connected:

\begin{theorem}				\label{thm61}

Assume the setting and notations of definitions \ref{def211} 
and \ref{def210}. 
In particular, suppose 
that $\ell$ is monotonically increasing. 
Then the following are equivalent:

{\rm (a)} $d$ satisfies the very weak $\na$-inequality

{\rm (b)} $\sigma$ satisfies the very weak $\na$-inequality

{\rm (c)} $\delta$ is increasing

{\rm (d)} the measure of intersection is increasing

{\rm (e)} $\sigma$ satisfies the $\Delta$-inequality

{\rm (f)} $\sigma$ satisfies the $\na$-inequality

Each of these assertions imply that $d$ satisfies the $\Delta$-inequality.

Reversely, the $\na$-inequality for $d$ implies the above assertions. 

\end{theorem}

\begin{proof}

(a) $\Leftrightarrow$ (b) is by lemma \ref{lem66}, 
(a) $\Leftrightarrow$ (c) by lemma \ref{lem65}, 
(c) $\Leftrightarrow$ (d)	by lemma \ref{lem61}, 
(c) $\Leftrightarrow$ (e)		by lemmas \ref{lem62} and \ref{lem63},
(c) $\Rightarrow$ 
(f) by lemma	\ref{lem64},
and 
(f) $\Rightarrow$ (b)	is clear. 
The 
forelast assertion is by lemma \ref{lem67}. 
That the last assertion 
implies (a) is clear. 
\end{proof}

Having said in definition \ref{def61} what the measure of an imagined intersection
of two sentences is, we continue to define it for the last Boolean 
operations of complement and union.

\begin{definition}[Measure of complement]

{\rm 
For all $x,y \in \ol S$ put 
$$
\zeta( x \backslash y): = \ell(x \na  
y) - \ell(y)$$

}
\end{definition}

As $\zeta( x \backslash y) = - \delta_y(x)$, $\delta$ is increasing if and only if the measure of complement $\zeta( x \backslash y)$ 
is decreasing in $y$. 
This is reminiscent of lemma \ref{lem12}: 

\begin{lemma}
If $\ell$ is monotone then
$$\zeta(x \backslash y)= d( y,x \nabla y)= \sigma(y,x \nabla y)$$
$$d(x,y)= \zeta(x \backslash y) \ud \zeta(y \backslash x)$$
$$\sigma(x,y)= \zeta(x \backslash y) + \zeta(y \backslash x)$$

\end{lemma}

\begin{definition}[Measure of union]

 {\rm 
For $x_1, \ldots,x_n \in \ol S$ set
$$\zeta(x_1 \cup \ldots \cup x_n):= \ell(x_1 \nabla \ldots \na 
x_n)$$

}
\end{definition}

The final goal of this section is to extend the measure of
imagined Boolean expressions of sentences to all 
possible Boolean expressions.  

\begin{definition}

Let $B$ be the free 
Boolean algebra (with zero element $0$ and one-element $1$) generated by the set $\ol S$.

\end{definition}

That is, $B$ consists of expressions like $x \cup y, x \cap y, x \backslash y = x \cap \ol y, \emptyset, \ol y$ (complement of $y$) etc. (also notated as $x + y, xy, x \ol y, 0, \ol y$) for
$x,y \in \ol S$ and $x,y \in B$. 

Two elements in $B$ are regarded as same 
if and only if they are the same by formal rules in 
Boolean algebra, for example, $x + x +y = x+y, x \ol x = 0, x + \ol x = 1$.


For $n=0$ and $x_i \in \ol S$ define the expression $x_1 \na \ldots \na x_n := \e$. 

The above definitions for $\zeta$ will now be generalized as follows: 
 
\begin{definition}		
		\label{def66}

Define a function 
 $$\zeta: B  \rightarrow \R$$
as follows. 
Put 
$$\zeta(0)=\zeta(1)= 0$$
For all $x_1,\ldots,x_n,y_1 , \ldots, y_m \in \ol S$ 
with $n \ge 0, m\ge 0$		
(and  $y:= y_1  \na  \ldots \na y_m$)
set
$$\zeta(x_1 \cap ... \cap x_n):= 
\sum_{k \ge 0,  
\{i_1,\ldots,i_k\}  \subseteq \{1,...,n\}} 
\ell(x_{i_1} \na \ldots \na x_{i_k}) (-1)^{k+1} $$

Here it is understood that mutually 
$i_s \ne i_t$.  
 Then define 
 \be{eq78}
\zeta(x_{1} \cap  ... \cap x_{n}  \cap \ol y_1  \cap ... \cap \ol y_m)
= \zeta ((x_{1} \cap  ... \cap x_{n}) \backslash (y_1 
\cup  ...  \cup y_m))
\en 
%
$$ : = \zeta(x_1 \cap ... \cap x_n) - \zeta(x_1 \cap ... \cap x_n \cap y)$$
\be{eq79}
= 
\sum_{k \ge 0,\{i_1,\ldots,i_k\}  \subseteq \{1,...,n\}} 
\ell(y \na x_{i_1} \na \ldots \na x_{i_k}) (-1)^{k+1} 
\en

because all expressions without $y$ cancel. 
Finally set 
$$\zeta (x \cup y): = \zeta(x) + \zeta(y) \qquad \mbox{if } x \cap y= 0$$ 

\end{definition}


\begin{lemma}
$\zeta$ is well defined.

\end{lemma}

\begin{proof}

Because we can choose common refinements, every expression in $B$ can be brought to the form $\sum_n z_n = \sum_n x_{1,n} \ldots x_{m_n,n} 
\ol y_{1,n} \ldots \ol y_{k_n,n}$ with $z_n z_m = 0$ for $n\neq m$,  
with $x_{i,j},y_{i,j} \in \ol S$ and $z_n$ obviously defined, 
and so the above $\zeta$ is already fully defined. 

This 
expression in $B$ is uniquely defined 
 if the set of involved variables is determined. Thus we only need
to check that the definition of $\zeta$ is independent of the 
set of involved variables:
  
Observe here the occurrence of $x_{n+1}$:
$$\zeta( x_{1}   ...  x_{n} x_{n+1}   \ol y_1   ...  \ol y_m) 
+ \zeta( x_{1}   ...  x_{n}  
\ol{x_{n+1}} \,\ol y_1   ...  \ol y_m )$$
%
$$= 
\sum_{k \ge 0,\{i_1,\ldots,i_k\}  \subseteq \{1,...,n+1\}} 
\ell(y \na x_{i_1} \na \ldots \na x_{i_k}) (-1)^{k+1} $$
$$+ 
\sum_{k \ge 0,\{i_1,\ldots,i_k\}  \subseteq \{1,...,n\}} 
\ell(y \na x_{n+1} \na x_{i_1} \na \ldots \na x_{i_k}) (-1)^{k+1} $$
$$= \zeta( x_{1}   ...  x_{n}   \ol y_1   ...  \ol y_m) $$
because all occurrences with $x_{n+1}$ cancel each other.  
This result is 
consistent 
with $x_{n+1} + \ol{x_{n+1}}= 1$. 

By induction we can add as many variables as we want by 
the formula $x + \ol x= 1$ and expansion,  
without changing $\zeta$. 
\end{proof}

\begin{lemma}					\label{lem610}
$\zeta$ is a ``signed measure'' 
on the Boolean algebra 
B generated by $\ol S$ 
in  the sense that 
$$\zeta(x \cap y)=\zeta(x) + \zeta(y) - \zeta(x \cup y)$$

for all $x,y \in B$. 
In particular we have
$$\zeta(\ol y)= - \zeta(y)$$
$$\zeta(x \backslash y)=  \zeta(x \cup y)   - \zeta(y) $$


\end{lemma}

\begin{proof}
This is clear, as when we form common refinements for $x$ and $y$
in $B$ as indicated in the last proof, $\zeta(x \cap y)$ 
sums  
only over the products 
which appear both in $x,y$,  and $\zeta(x \cup y)$ 
over 
those which appear in $x$ or $y$.

The last two identities follow from this and $\zeta(1)=0$. 
\end{proof}

We 
should however be reminded that we do not work with a real positive measure 
on measurable sets.  We only observe:

\begin{lemma}

For all $x,y,x_i \in \ol S$ we have: 

(i)
$\zeta(x \cap y) \ge 0$, $\zeta(x_1 \cup \ldots \cup x_n) \ge 0$ 

(ii)
If $\ell$ is monotonically increasing then $\zeta(x\backslash y) \ge 0$.  


\end{lemma}

%

\section{New metric}			

									\label{sec7}

Up to now we could not clarify if $d$ or $\sigma$ satisfies
the triangle or $\na$-inequality in our natural prototype languages. 
We have theorem \ref{thm61}, and it tells us the 
$\na$-inequality 
of $d$ implies all the desired natural properties. 
In this section we show how we may modify a given length function $\ell$ such that we get the 
$\na$-inequality for its associated distance function $d_\ell$ and thus all the desired properties.  
The approach is 
as follows. At first we modify the distance function
$d_\ell$ for a given $\ell$ as in definitions \ref{def71}
and \ref{def72} such that the so obtained $d$ satisfies
the $\Delta$- and $\na$-inequality. 
Then we extract from it a new length function $\ell_2$ by
setting $\ell_2 (x):=d(x,\e)$ (see definition \ref{def53}). 
But this new $\ell_2$ may not be monotone increasing 
as required in (\ref{eql3}). That is why we make it monotone
increasing in definition \ref{def74} to obtain a new length
$\ol{\ell_2}$. 
But the validity of the desired $\Delta$- and 
$\na$-inequality 
of $d_{\ol{\ell_2}}$ and $\sigma_{\ol{\ell_2}}$ associated to $\ol{\ell_2}$
might fail,  also cf. lemma 
\ref{lemma713}.
Hence we restart the above procedure
with $\ol{\ell_2}$ again, and then again and again infinitely 
many often, and in the limit the obtained length function
$\ell^\infty$ satisfies all desired properties.
 
 A key lemma for this procedure is lemma \ref{lem511}, whose proof shows that we need both the $\Delta$- and $\na$-inequality simultaneously. Hence, it is also intrinsic to do the procedure with the $d$-distance, because the lemma might not hold for the 
 $\sigma$-distance. 
 Doing the procedure for the $\sigma$-distance instead leads at least to the estimates of lemma \ref{lem727} in the limit.
 

 Actually, 
in this section we start with $d$, and if nothing else is said, $d$ always denotes 
a symmetric, positive, nilpotent function $d: \ol S \times \ol S \rightarrow \R$ for 
an abelian, idempotent monoid $(\ol S,\na,\e)$. 
%

Example \ref{ex106} suggests, that $d$ of definition \ref{def210}
does not satisfy the $\Delta$-inequality in general. 
A way out of that fact is to repair
$d$ as follows:



\begin{definition}[Forced $\Delta$-inequality]						\label{def71}

Define a function $\tilde d: \ol S \times \ol S \rightarrow \R$ by
$$\tilde d(x,y) = \inf_{n \ge 0, a_1,\ldots,a_n \in \ol S}  d(x,a_1) + d(a_2,a_3) + \ldots + d(a_{n-1},a_n) + d(a_n,y)$$

\end{definition}


\begin{lemma}			\label{lemma51}

$\tilde d$ satisfies the $\Delta$-inequality, and $d = \tilde d$ if and only if $d$ satisfies the $\Delta$-inequality.


\end{lemma}

\begin{proof}

This is obvious as the right hand side of the $\Delta$-inequality 
is one
possible value over which the infimum is taken in $\tilde d$. 
%
\end{proof}

The following lemma may be useful for the pseudometric
$\rho$ on $\ol S$ defined by $\rho(x,y) := |\ell(x)  - \ell(y)|$ 
to get a lower bound for $\tilde d$, cf. lemma \ref{lem34}.

\begin{lemma}				\label{lemma52} 
If 
a function $\rho: \ol S \times \ol S \rightarrow \R$ satisfies the $\Delta$-inequality then  we
have the following lower bound estimate: 
$$\rho \le d \quad \Rightarrow \quad \rho \le \tilde d$$


\end{lemma}

Next we repair $d$ to satisfy the $\na$-inequality:

\begin{definition}[Forced $\na$-inquality]					\label{def72}

Define a function $\hat d: \ol S \times \ol S \rightarrow \R$ by
$$\hat d(x,y) := \inf_{x  = 
x_1 \nabla \ldots \nabla x_n, \,  y  = 
y_1 \nabla \ldots \nabla y_n, \; x_i,y_i \in \ol S}
	d(x_1,y_1) + \ldots  +  d(x_n,y_n)$$

\end{definition}

\begin{lemma}				\label{lemma54}

$\hat d$ satisfies the $\na$-inequality, and $d = \hat d$ if and only if $d$ satisfies the $\na$-inequality.


\end{lemma}

\begin{proof}

This is obvious as the right hand side of the 
$\na$-inequality is one
possible value over which the infimum is taken in $\hat d$.
\end{proof}

Fortunately, the transition from $d$ to $\tilde d$ does 
not destroy the validitiy of the $\na$-inequality:

\begin{lemma}					\label{lemma55}

If $d$ satisfies the $\na$-inequality, then $\tilde d$ 
satisfies the $\na$-inequality.




\end{lemma}

\begin{proof}

By application of the $\na$-inequality for $d$ we get 
$$\tilde d(x \na y ,z \na w) = \inf_{n \ge 0, a_1,\ldots,a_n \in \ol S}  d(x \na y,a_1) + d(a_2,a_3) + \ldots  + d(a_n, 
z \na w)$$
$$\le \inf_{n \ge 0, a_i, b_i \in \ol S}  d(x \na y,a_1 \na b_1) + d(a_2 \na b_2 ,a_3 \na b_3) + \ldots  + d(a_n \na b_n, 
z \na w)$$
$$\le \tilde d(x,z) +  \tilde d(y,w)$$ 

where the first inequality is because the infimum goes 
over a smaller selection. 
\end{proof}


\begin{lemma}				\label{lemma562}

$\tiha d$ satisfies the $\na$- and $\Delta$-inequality, and $d = \tiha d$ if and only if $d$ satisfies the $\na$- and $\Delta$- inequality. 

\end{lemma}

\begin{proof}

The 
first  assertion  follows 
by 
lemmas \ref{lemma51}, \ref{lemma54} 
and \ref{lemma55}.

If $d = \tiha d$ then $d \ge \hat d \ge \tiha d = d$, 
thus $d = \hat d \ge \tiha d = d = \hat d$, so that the last equivalence is by lemmas \ref{lemma55} and \ref{lemma51}. 
\end{proof}


Hence we get a projection from metric candidates to those
ones which satisfy the $\Delta$- and $\na$-inequality: 

\begin{corollary}[Projection onto pseudometrics satisfying $\na$-inequality]

There is  a projection $d \mapsto \tiha d$ 
from the set of positive, symmetric, nilpotent functions on $\ol S$
onto the set of pseudometrics on $\ol S$ satisfying the $\na$-inequality. 

\end{corollary}

The following bar operator turns 
possibly non-increasing length functions to increasing
ones: 

\begin{definition}[Bar operator]				\label{def74}

Given a positive function $\ell : \ol S \rightarrow \R$ set 
$\ol \ell : \ol S \rightarrow \R$ as 
$$\ol \ell(x):= \inf_{z \in \ol S} \ell( x \na z)$$ 

\end{definition} 

\begin{lemma}						\label{lem76}
If $\ell: \ol S \rightarrow \R$ is a possibly non-increasing length function then 
$\ol \ell$ is an increasing length-function, i.e. 
(for all $x,y \in \ol S$)
$$ \ol \ell(x) \le \ol \ell(x \na y) $$  

If $\ell$ is already monotonically increasing, then $\bar \ell = \ell$. 
\end{lemma}

From the various distance functions we naturally extract the length functions 
out:

\begin{definition}			\label{def53}

{\rm

Define positive functions $\ell,\tilde \ell, \hat \ell, \tiha \ell : \ol S \rightarrow \R$  (depending on $d$)  by
%
\be{eq167}
\ell(x):= d(x,\e), \;\tilde \ell(x):= \tilde d(x,\e), \;\hat \ell(x):= \hat d(x,\e), \; \tiha \ell(x):= \tiha d(x,\e)
\en

}
\end{definition} 

\begin{lemma}						\label{lem77}

If $d$ satisfies the $\na$-inequality, then $\ell$ of \re{eq167} is a non-increasing length function, and $\ol \ell$ a length function. 

\end{lemma}

The next lemma shows that switching from $d$ to $\tiha d$
may really change their associated length functions $\ell$ and
$\tiha \ell$, respectively:

\begin{lemma}				\label{lemma713} 

If $\rho \le d$ 
for $\rho$ as defined before lemma \ref{lemma52}
then 
$\tilde \ell = \ell$, and 
if $\ell$ is subadditive then $\hat \ell = \ell$. 
But we  have 
$\tiha \ell \neq \ell$ 
in certain cases. 



\end{lemma}

\begin{proof}
By lemma 
\ref{lemma52} we have
$$|\ell(x) - \ell(\e)| \le \tilde d(x,\e) \le \ell(x)$$

If $\tiha \ell = \ell$ then  $d_\ell$
would satisfy the $\Delta$-inequality by lemma \ref{lem58}.  
(By independence, the forward reference is allowed.) 
But we know from 
example \ref{ex106} that this is not in general true. 
\end{proof}

This is an important key lemma which shows how our 
above procedure decreases the metric candidates:
 
\begin{lemma}			\label{lem511}
If $d$ satisfies the $\Delta$- and $\na$-inequality then
for $\ell$ of \re{eq167}, 
$$d_\ell \le d$$
$$d_{\ol \ell} \le d$$

\end{lemma}

\begin{proof}
(i) Without loss of generality we have by the second $\Delta$-inequality and the $\na$-inequality that 
$$d_\ell(x,y)=  d(x \na y,\e) - d(x,\e) $$
$$\le d(x \na y,x) 
\le d(x,y)  + d(x,x)= d(x,y)$$
 
(ii)
Similarly 
we get 
$$d_{\ol \ell}(x,y) = \inf_z  \ell(x \na y \na z)  - 
\inf_w  \ell(x \na w) \ad \inf_w   \ell(y \na w)$$

Given 
any $\varepsilon >0$ we may choose a $w \in \ol S$ such that 
this is 
$$
\le \inf_z  \ell(x \na y \na z)  -    \ell(x \na w)  + \varepsilon$$
$$\le   \ell(x \na y \na w)  -   \ell(x \na w)  + \varepsilon$$
$$\le  d( x \na y \na w, x \na w)  + \varepsilon $$
$$\le  d(y,x)
+ d(x \na w, x \na w) + \varepsilon= d(x,y)+ \varepsilon$$


\end{proof}

\begin{corollary}				\label{cor51}

Using definition \re{eq167},   we have 
$$d_{{\tiha \ell} } \le \tiha d $$   
$$d_{\ol {\tiha \ell} } \le \tiha d $$    

\end{corollary}

\begin{proof}

Apply lemma \ref{lem511} to $\tiha d$. 
%
%
\end{proof}

Let us now restart the above procedure with a new 
$ d_{\tihabar {\ell}}$ instead of $d$ and then again and again. 
More precisely put


\begin{definition}					\label{def75}
Set a start point 
$d^{(1)} := d:\ol S \times \ol S \rightarrow \R$,  and then 
for all $n \ge 1$ define successively
functions $ \tiha {\ell^{(n)}},\ell^{(n+1)}: \ol S \rightarrow \R$
and $d^{(n+1)}: \ol S \times \ol S \rightarrow \R$  by 
%
$$ 
\tiha {\ell^{(n)}}(x):= \tiha{ d^{(n)}}(x,\emptyset)
$$
%
%
$$\ell^{(n+1)} := {\ol {\tilde {\hat {\ell^{(n)}}}}}$$
%
%
%
\be{eq200}
d^{(n+1)}:= d_{\ell^{(n+1)}}
\en  

\end{definition} 

We are going to show that the successively applied procedure 
defined in the last definition which takes a length function in and gives a length function out converges to a desired 
fixed point.

\begin{lemma}			\label{lemma510}

All $\ell^{(n+1)}$ 
are length functions on $\ol S$, and 
we have descending sequences 
$$0 \le ...  \le \ell^{(3)} = \ol {\tilde {\hat {\ell^{(2)}}}} 
\le {\tilde {\hat {\ell^{(2)}}}}  \le \ell^{(2)} = 
\ol {\tilde {\hat \ell}} \le
\tilde {\hat \ell} \le \ell $$
$$0 \le ...  \le d^{(3)} = d_{\ell^{(3)} }  \le \tilde {\hat {d^{(2)}}}  \le d^{(2)} 
= d_{\ell^{(2)} } \le \tilde {\hat d} \le d $$


\end{lemma}

\begin{proof}

By lemmas \ref{lemma562} and 
and \ref{lem77}, $\ell^{(n+1)}$ is a length function. 
By definitions \ref{def71} and \ref{def72}, we have $\tiha d \le d$. 
By 
corollary \ref{cor51} applied to $d:= d^{(n)}$ and $\tiha{\ell}:= 
\tiha {\ell^{(n)}}$ 
 we thus obtain $d^{(n+1)} 
= d_{\ell^{(n+1)}} 
= d_{\tihabar{ \ell^{(n)}}} 
\le \tiha {d^{(n)}} \le {d^{(n)}}$. 
Hence, as $\ol \ell \le \ell$ by 
definition \ref{def74}, 
$$\ell^{(n+1)} = \ol {\tiha {\ell^{(n)}}} \le  
{\tiha {\ell^{(n)}}} = \tiha {d^{(n)}}(-,\e) \le d^{(n)}(-,\e) = 
d_{\ell^{(n)}}(-,\e) = \ell^{(n)}$$
\end{proof}

We note that all $\tiha {d^{(n)}}$ are pseudometrics
which satisfy the $\na$-inequality by lemma \ref{lemma562}, but it is unclear
whether they are of the form $d_{\ell'}$. (Otherwise we would be done.)


Since we have monotonically falling real-valued  sequences that are bounded below by $0$ we may define
 
\begin{definition}			\label{def56}
Define the pointwise function limits
$$\inft d:= \lim_{n \rightarrow \infty} d^{(n)}$$
\begin{equation}		\label{eq6}
\inft \ell:= \lim_{n \rightarrow \infty} \ell^{(n)}
\end{equation}

\end{definition}

\begin{lemma}						\label{lem710}

The function 
$$d^{(\infty)} = d_{\inft \ell} = \lim_{n \rightarrow \infty} d_{\ell^{(n)}} = 
\lim_{n \rightarrow \infty} \tiha {d_{\ell^{(n)}}} = \tiha {d^{(\infty)}} = \tiha {d_{\inft \ell} }$$
%
is a pseudometric on $\ol S$ satisfying the $\na$-inequality. 



\end{lemma}

\begin{proof}
The second equality is clear by the formula for $d_\ell$ of def. \ref{def210}, the first
one is then by definition.
 The third equality is because of the converging sequence of 
 lemma \ref{lemma510}. 
 Because $\tiha {d_{\ell^{(n)}}}$ satisfies the $\na$- and
 $\Delta$-inequality by lemma \ref{lemma562},  their point-wise limit does also, and so 
 we get the last equalities and thus the final assertion.
 \end{proof}

\begin{lemma}			\label{lemma513}

$\inft \ell$ is a length function on $\ol S$ and 
 %
$$ \ell^{(\infty)} = {d_{\inft \ell}}(-,\e)=  {\tilde {\hat {\ell^{(\infty)}}}}
:= \tiha{d_{\inft \ell}}(-,\e) = \ol {\tilde {\hat {\ell^{(\infty)}}}} = \ol {\ell^{(\infty)}}  $$


\end{lemma}

\begin{proof}


%
As $\tiha {d_{\inft \ell}} = d_{\inft \ell}$ by lemma \ref{lem710}, we get 
the second identity. 
%

 
 The last equality holds because, as  $d^{(\infty)} = d_{\inft \ell}$ is positive,
 $${d_{\inft \ell}}( x\na y,x)= \inft \ell(x \na y) - \ell^\infty(x) \ge 0$$

that is, $\inft \ell$ satisfies already the monotone 
increasingness. 
%
  %
\end{proof}

We now 
summarize that the successively applied described procedure which takes a length function 
 and forms a new one converges to a desired fixed point: 

\begin{theorem}[Ideal length function]				\label{cor52}


Given a positive symmetric idempotent function $d$ on an abelian, idempotent monoid $(\ol S,\na ,\e)$, 
define $ \inft \ell$ as described in definitions \ref{def75}
and \ref{def56}.

Then 
$\inft \ell$ is a length function on $\ol S$, 
and $d_{\inft \ell}$ and $\sigma_{\inft \ell}$  (def. \ref{def210}) 
are pseudometrics on $\ol S$ 
satisfying the $\na$-inequality. 
%



\end{theorem}


\begin{proof}

By lemmas \ref{lem710} and \ref{lemma513}, and theorem \ref{thm61}. 
%
%
\end{proof}

If $d$ starting at the beginning of this section happens to be of the form $d=d_\ell$ for a length function $\ell$ on $\ol S$, then
astonishingly $d$ satisfies already the $\Delta$-  and
$\na$-inequality if only $\ell = \tiha \ell$:

\begin{proposition}[Ideal length function attained]				\label{lem58}


If $d$ is of the form $d=d_\ell$ for some length function $\ell$ then 
$d_\ell $ satisfies the $\na$- and $\Delta$-inequality 
$\Leftrightarrow$ 
$$d_\ell = \tiha{ d_\ell} \quad \Leftrightarrow 
\quad 
\tiha \ell = \ell		
%
\quad \Leftrightarrow 
\quad 
\tihabar \ell = \ell		\quad 
\Leftrightarrow 
\quad 
 \ell = \inft \ell$$

\end{proposition}
 

\begin{proof}
If $d$ satisfies the $\Delta$-inequality then certainly
$d \le \tilde d \le d$ and similarly 
the $\na$-inequality 
implies $d = \hat d$. So $\tiha d =d$ by lemma \ref{lemma55} and thus $\tiha \ell = \ell$. 

If $\tiha \ell = \ell$ then setting $d:=d_\ell$ we have with 
corollary \ref{cor51}
that 
$$d_\ell = d_{\tiha \ell} \le \tiha d = \tiha{d_\ell} \le d_\ell$$
and so $d_\ell = \tiha{d_\ell}$ 
satisfies the $\na$ and $\Delta$-inequality. 

If $\tihabar \ell = \ell$ then 
$$\tihabar \ell = \ell \ge \tiha \ell \ge \tihabar \ell = \ell$$
and so $\ell = \tiha \ell$. 
Also then, $\ell^{(2)} = \tihabar \ell =\ell$, and thus $d^{(2)}= d_\ell$
 and 
 the above procedure restarts again with the same start $d_\ell$, 
 and so 
 again $\ell^{(3)}=\ell^{(2) 
 }$ etc and we get a constant sequence
and finally 
$\inft \ell = \ell$. 

If $\inft \ell = \ell$ then by corollary \ref{lemma513} $\tihabar \ell = \ell$. 
\end{proof}

We just remark that one 
will naturally turn the desired pseudometrics to metrics effortlessly: 

\begin{corollary}			

By  definition \ref{def33},
we
may then divide out elements in $\ol S$ to obtain 
a metric space 
$\me S$, which is again an abelian idempotent monoid, 
 with metrics $ \me {d_{\inft  \ell}} = d_{\me {\inft \ell}}$ 
and 
$ \me {\sigma_{\inft  \ell}} = \sigma_{\me {\inft \ell}}$ 
satisfying the $\na$-inequality.  

\end{corollary}				

\begin{proof}

By theorem \ref{cor52} we have pseudometrics 
$\inft d$ and $\inft \sigma$, 
and 
to these we apply definition \ref{def33}, proposition 
\ref{lem36}  and lemma \ref{lem35}. 
%
%
%
\end{proof}

The above procedure 
gives a desired projection onto the set of length functions 
from arbitrary length functions to desired ones:

\begin{corollary}[Projection onto ideal length functions]			\label{cor724}

There is 
a projection $\ell \mapsto \inft \ell$ from the set 
of length functions $\ell$ on $\ol S$ 
onto 
the set of those ones 
such that $d_{\ell}$ and $\sigma_{\ell}$
are pseudometrics which satisfy the $\na$-inequality. 

\end{corollary}

\begin{proof}

Start the above procedure, definition \ref{def75}, with $d:=d_\ell$ (or alternatively 
with $d:=\sigma_\ell$, which however may result in another fixed point $\inft \ell$)  
and end up with 
theorem \ref{cor52}. 

Projection is onto  
because of lemma \ref{lemma562}  
and theorem \ref{thm61}. 
%
\end{proof}

If the 
length function $L$ on $S$ is a simple character count as in our natural protoexamples 
then the fixed point for the length functions will be attained 
after finitely many steps: 

\begin{lemma}

If the $\ell$-function is in a grid, i.e. 
$\ell(\ol S) \subseteq \gamma \Z$ for some constant $\gamma \in \R$, and $d= d_\ell$, then 
the limits 
of definition \ref{def56} 
will all be attained at some $n$ ($n$ 
depending on the points 
of $\ol S$ however).  

\end{lemma}

We will now make a bigger cut and switch from $d$ to $\sigma$. 
Let us now do the above procedure by  
starting with a length function $\ell^{(1)}$ on $\ol S$, 
but instead always forming $d_{\ell^{(n)}}$ we now choose
$\sigma_{\ell^{(n)}}$. In other words, set

\begin{definition}				\label{def77}

Do everything as in 
definition \ref{def75} but modify \re{eq200}
to 
$$d^{(n+1)}:= \sigma_{\ell^{(n+1)}}$$

\end{definition} 

\begin{lemma}				\label{lem727}

Assume 
definition \ref{def77}, but to avoid confusion with 
what we have defined before, rewrite $\sigma^{(n)}:= d^{(n)}$. 

Then $(\ell^{(n)})_n$ is a decreasing sequence of length functions on $\ol S$, with limits  and estimates 
$$\inft \ell := \lim_{n \rightarrow \infty} \ell^{(n)}, \qquad \inft \sigma := \lim_{n \rightarrow \infty} \sigma^{(n)}
= \sigma_{\inft \ell}$$
\be{eq90}
1/2  {\sigma^{(n +1)}} \le  \tiha{\sigma^{(n)}}  \le  {\sigma^{(n)}} 
\en
%
\be{eq91}
1/2 
\inft \sigma  \le 
\liminf_{n} \tiha{\sigma^{(n)}}  \le 
\limsup_{n} \tiha{\sigma^{(n)}}   \le 
\tiha{ \inft \sigma}   \le   
\inft \sigma 
\en
\end{lemma}

\begin{proof}

As $\ell^{(n)}$ converges pointwise to $\inft \ell$, 
$\sigma^{(n+1)}= \sigma_{\ell^{(n+1)}}$ converges also pointwise to $\inft \sigma$.    
\re{eq90} is by lemmas \ref{lem34} and
corollary \ref{cor51} applied to $d:= \sigma^{(n)}$, as  
$$1/2 \sigma_{\ell^{(n+1)}} 
= 1/2 \sigma_{\tihabar{\ell^{(n)}}}
\le d_{\tihabar{\ell^{(n)}}} \le 
\tiha {\sigma^{(n)}}$$

By 
definitions  \ref{def71} and \ref{def72}, 
and because $\sigma^{(n)}$ converges pointwise to $\inft \sigma$, 
for all $\varepsilon >0$ there are $a_1,a_2, x_0,a_{1,0}, \ldots \in \ol S$
and 
$n_0 \ge 1$ such that 
for all $n \ge n_0$ we have that 
$$\tiha {\inft \sigma}(x,y)   + 3\varepsilon  \ge 
 \hat {\inft \sigma}(x,a_1) 
+ 
	\hat {\inft \sigma}(a_1,a_2) + \ldots + 
	\hat {\inft \sigma}(a_m,y)  
	+ 2\varepsilon $$
$$ \ge {\inft \sigma}(x_0,a_{1,0}) 
	+ \ldots + 
	 {\inft \sigma}(a_{m,k_m}',y_m)  
	+\varepsilon
	$$
$$ \ge {\sigma^{(n)}}(x_0,a_{1,0}) 
	+ \ldots + 
	{\sigma^{(n)}}(a_{m,k_m}',y_m)  
	$$
$$ \ge \tiha {\sigma^{(n)}}(x_0,a_{1,0}) 
	+ \ldots + 
	\tiha {\sigma^{(n)}}(a_{m,k_m}',y_m)  
	\ge \tiha {\sigma^{(n)}}(x,y)
	$$

where the last inequality is 
because $\tiha \sigma^{(n)}$ 
satisfies the $\na$- and $\Delta$-inequality and so we 
can estimate back what we have expanded in the first two 
lines above.

This and \re{eq90} yield \re{eq91}. 
\end{proof}

\section{Homomorphisms}

In this section we shift our attention away from the abelian
idempotent monoids to 
their homomorphism sets between them. 
Thereby we consider only homomorphisms which are ``uniformly
continuous'' similarly as the bounded linear operator between
Banach spaces are, see definition \ref{def85}. 
We are provided by a length on each homomorphism set
(analogy: norm on operators between Banach spaces) 
and apply the procedure of 
section \ref{sec7} to obtain 
new length functions with better theoretical properties. 
Thereby we enrich this concept with the desired property 
of product inequalities, that is, we want achieve the validity 
of inequalities like $\ell(U \circ V) \le \ell(U) \ell(V)$ and
$d(A \circ U, A \circ V) \le \ell(A) d(U,V)$ 
for homomorphisms $A,U$ and $V$. 
Similarly as in the last section we end up with a fixed point
theorem yielding desired length functions, and apply them 
to get a Banach-Mazur-like distance between languages.

At first we define homomorphism between languages. 
In this section let
$(S,\equiv_S)$ and $(Q,\equiv_Q)$  be languages with equivalence relation as in definitions \ref{deflangua} and 
\ref{defequ}. 

\begin{definition}[Homomorphism of languages]				\label{def81} 

{\rm 

A
function $T: (S,\equiv_S) \rightarrow (Q,\equiv_Q)$ 
is called a homomorphism between these languages 
if
 for all $x,y \in S$ 
 we have  
$$x \bot y \quad \Rightarrow  
\quad  \Big ( Tx \bot Ty   \quad 
\mbox{and} \quad T(x \ud y) = 
Tx \ud Ty  \Big )$$
 $$x \equiv_S y 
 \quad \Rightarrow  
\quad  Tx \equiv_Q Ty $$

} 
\end{definition}

In this section we almost exclusively return 
to the abstract setting of abelian idempotent monoids. 
In what follows let $(\ol S,\na , \e),(\ol Q,\na ,\e)$ and $(\ol R,\na,\e)$ denote 
any 
abelian idempotent monoids.
Finally, a natural notion for homomorphisms between 
such monoids is certainly a semigroup or monoid homomorphism:
%

\begin{definition}			\label{def82}

{\rm 

A function 
$T: \ol S \rightarrow \ol Q$ 
is called a homomorphism
if  for all $x,y \in \ol S$ 
$$T(x \na y) =  Tx \na Ty$$

}
\end{definition}

A homomorphism between languages becomes 
a homomorphism between their quotients: 

\begin{lemma}

If $T:(S,\equiv_S) \rightarrow (Q,\equiv_Q)$ is a homomorphism then it descends 
to a homomorphism 
$\ol T: \ol S \rightarrow \ol Q$ defined by
$\ol T([s])= [T(s)]$. 

\end{lemma}

The monoid operation $\na$ between abelian idempotent
monoids naturally carries over to the set of homomorphisms between them:

\begin{definition}[Operation on homomorphism set]					\label{def83}

{\rm 

For homomorphisms $U,V:\ol S \rightarrow \ol Q$ 
define the homomorphism 
$$U \na V : \ol S \rightarrow \ol Q: (U\na V)(x) := U x \na V x$$

}
\end{definition}

That is why we get the usual order on the homomorphism set:

\begin{definition}

{\rm 

We equip $\hom(\ol S,\ol Q)$ with the usual order $\le$
as of definition \ref{deforder}, 
 that is,  for all homomorphisms $U,V: \ol S \rightarrow \ol Q$ 
 set
$$U \le V \quad \Leftrightarrow \quad 
V = U \na V
 \quad \Leftrightarrow \quad
 Ux \le Vx \; \forall x \in \ol S$$

} 
\end{definition}

\begin{lemma}


For any $a \in \ol S$, $T_a: \ol S \rightarrow \ol S$ defined by
$T_a(x) =x \na a$ is a (non-unital) homomorphism .

\end{lemma}

We write the composition $U \circ V$ of two homomorphisms $U,V$
also as a product $U V$. 
The following lemma states the distributive law 
between the (``additive'') monoid operations and the (``multiplicative'') composition operations on the homomorphism  
sets.


\begin{lemma}[Distribution laws]				\label{lemma83}
For all homomorphisms $U,V :\ol Q \rightarrow \ol R$ and 
$X :\ol S \rightarrow \ol S$ and 
$Y :\ol R \rightarrow \ol W$ 
we have 
$$(U \nabla V)  X= U  X \nabla  V  X$$
$$Y  (U \nabla V) = Y U  \nabla  Y  V $$

\end{lemma}

Write $\e$ for the operator $\e:\ol S \rightarrow \ol T$ $s \mapsto [\e]$. 
A continuous  linear operator between 
Banach spaces  
is actually {\em uniformly continuous}, that is, is bounded.
We consider analogously only uniformly continuous, or in other words, bounded homomorphisms between abelian idempotent 
monoids equipped with metrics 
in order to get a natural length function on each homomorphism set 
in analogy to the norm of linear operators between Banach spaces: 

\begin{definition}[Natural length function on homomorphism set]				\label{def85}

{\rm 

Let $\ol S$ and $\ol Q$ be equipped with 
positive, symmetric, nilpotent functions 
$d_{\ol S}$ and $d_{\ol Q}$, respectively,     
satisfying the $\na$-inequality.  
Write $\hom(\ol S , \ol Q)$ for the set of all 
``uniformly continuous'' homomorphisms 
$U: \ol S \rightarrow \ol Q$,  
i.e. for which $\ell'(U)$ defined next is finite.  
%
 
We define functions $\ell',\ell: \hom(\ol S,\ol Q) \rightarrow \R$, 
where $\ell$ is a length function,  
by 
$$\ell'(U):=  \inf \{ M \ge 0|\, M \in \R,\;\forall x,y \in \ol S: \; d_{\ol Q}(Ux,Uy) \le M  d_{ \ol S}(x,y) 
\}$$
%
%
\be{eq110}
\ell := \ol{\ell'}
\en


}
\end{definition}

In \re{eq110} we used the bar operator of definition \ref{def74}. 
Actually we should 
index in $\ell$ the data $\ol S,\ol Q$ etc., but to avoid cluttering we use the sloppy light notation. 
If we choose for $d_{\ol Q}$ the trivial null function then we will 
get all homomorphisms in $\hom(\ol S, \ol Q)$.   
For simplicity, we shall call elements of $\hom(\ol S, \ol Q) $
just homomorphisms and drop ``uniformly continuous''.

\begin{lemma}							\label{lem84}
For all (meaningful composable) homomorphisms $U,V$ we have 
\be{eq168}
\ell'(U \circ V) \le \ell'(U)  \ell'(V)
\en
%
%
$$\ell'(U \na V) \le \ell'(U ) + \ell'(V)$$
%



\end{lemma}

\begin{proof}

The first inequality is elementary as for the norm of linear operators,
and the second one follows from the $\na$-inequality of $d_{\ol Q}$.
%
\end{proof}

We refer to the inequality \re{eq168} (for any given $\ell'$) as the {\em product-inequality}. 
  It is analogue to the product inequality $\|S T\|\le \|S\| \|T\|$
  for bounded linear operators between Banach spaces. 
We tidy up the 
situation so far and get the usual abstract setting
of abelian idempotent monoids also for the homomorphism sets:

\begin{lemma}[Abstraction]						\label{lem85}

The homomorphism set 
$\big((\hom(\ol S, \ol Q),\na , \e \big)$ is an abelian idempotent monoid and $\ell$ of definition \ref{def85} a length function on it.

\end{lemma}

\begin{proof}

By definition \ref{def83} we 
check the first assertion, by lemma  \ref{lem84}, $\ell'$ is a 
non-increasing length function, to which we apply lemma 
\ref{lem76}. 
\end{proof}

For the further discussion we need now consider a whole
family of abelian idempotent monoids:
 
\begin{definition}[Category of languages]		\label{def86}		

{\rm 

Let 
$\Lambda$ be a 
small category 
consisting of abelian, idempotent monoids $(\ol S,\na,\e )$ equipped
with positive, symmetric, nilpotent functions 
$d_{\ol S}:\ol S \times \ol S \rightarrow \R$ 
satisfying the $\na$-inequality 
as objects,  
and all possible ``uniformly continuous'' homomorphisms between these objects as in definition \ref{def82} 
with ordinary composition  
as morphisms. 

}
\end{definition}

It is now important to stress that we will in the sequel not use the distance functions 
$d_{\ol S}$ on the objects 
at all. 
We only mentioned it to 
have a natural length function on the homomorphism sets 
by definition \ref{def85}. 
We require now all homomorphism sets to be equipped with length functions and metric candidates:

\begin{definition}[Category of languages with general length functions on homomorphism sets]					
\label{def87}	

{\rm

Each morphism set $M$ ($= \hom(\ol S, \ol Q)$) in $\Lambda$  
assume to be equipped with a length function $\ell: M \rightarrow \R$ 
  and a positive, symmetric, 
nilpotent  function 
$d:M \times M \rightarrow \R$. 
(For example the length functions $\ell$ of definition \ref{def85}, and  $d:= d_\ell$ 
of definition \ref{def210} 
are the typical natural choices.) 
For simplicity, all $\ell$ and $d$ are called in the same way, independent of $M$. 


}
\end{definition}

Our next aim is to modify the 
length functions on the homomorphism sets 
in such a way such that we (just) get the natural inequality
$d(A \circ U, A \circ V) \le \ell(A) d_\ell(U,V)$ for all composable homomorphisms
$A,U,V$ in analogy to the inequality
$\|AU - AV \| \le \|A\| \|U - V\|$ for linear operators
between Banach spaces. 
Simultaneously we alter $\ell$ as to achieve that
$d_\ell$ satisfies the $\Delta$- and $\na$-inequality.  
To this end we apply the fix point concept of section
\ref{sec7} to the homomorphism sets, enriched by the now also desired product
inequality.

For the rest of this section we will define various new $\ell$s and
$d$s simultaneously for all morphism sets $M$ in 
$\Lambda$, 
without indicating $M$ in notation. 
We again note that we may choose all functions $d_{\ol S}$ to be the trivial null functions and 
any $\ell$ and $d$s in definition \ref{def87}	 independent  of the 
suggested $\ell$s of definition \ref{def85}. 

We repair now all given symmetric nilpotent positive functions 
$d$ on the homomorphism sets in such a way 
that they satisfy the product inequalities with respect to the
length functions: 

\begin{definition}[New submultiplicative distance function]						\label{def88}

For all objects $\ol S$ and $\ol Q$ in $\Lambda$ define  $\dot d : \hom(\ol S,\ol Q) \times \hom(\ol S,\ol Q) \rightarrow \R$ 
as follows.
For all morphisms 
$X,Y: \ol S \rightarrow \ol Q$ in $\Lambda$ 
with same domain and range
set 
%
 %
$$\dot d (X,Y) := \inf_{ 
\stackrel{X = A_1 ... A_n U B_1 .... B_m, \;}
{Y = A_1 ... A_n V B_1 .... B_m}, \; n, m\ge 0}
\ell(A_1) ...  \ell(A_n)  d( U,V)  \ell( B_1 )....   \ell(B_m)
$$

where $A_i, V,U,B_i$ are 
morphisms between any objects of $\Lambda$ 
for which their products as indicated are valid (composable).

\end{definition}

The so repaired functions $\dot d$ now fulfill the
desired product inequalities: 

\begin{lemma}[Product inequality]  \label{lemmulti}

For composable morphisms $A,B,X$ in $\Lambda$ 
we have the submultiplicativity relations 
(often called 
{product-inequality})  
\begin{equation}		\label{pe1}
\dot d(A X,B X) \le \dot d(A,B) \ell(X)
\end{equation}
%
\begin{equation}		\label{pe2}
\dot d( X A, X B ) \le  \ell(X) \dot d(A,B) 
\end{equation} 

\end{lemma}

\begin{proof}

Inserting in definition \ref{def88}
 we have  
$$\dot d (Q X,  R X ) := \inf_{ \stackrel{QX = A_1 ... A_n U B_1 .... B_m, \;}
{RX = A_1 ... A_n V B_1 .... B_m}}
\ell(A_1) ...  \ell(A_n)  d( U,V)  \ell( B_1 )....   \ell(B_m)
$$
$$\le \inf_{ \stackrel{Q = A_1 ... A_n U B_1 .... B_m, \;}
{R = A_1 ... A_n V B_1 .... B_m}}
\ell(A_1) ...  \ell(A_n)  d( U,V)  \ell( B_1 )....   \ell(B_m) \ell(X)
$$

because we choose the infimum over a smaller set in the second line, that is, we at first chose a resolution $Q = A_1 ... A_n U B_1 .... B_m$ and then said $QX A_1 ... A_n U B_1 .... B_m X$. 
\end{proof}


The validity of relations (\ref{pe1}) and (\ref{pe2}) for all morphism sets of $\Lambda$
is summarized as {\em product inequality}.
It is important that $\dot d$ descends: 

\begin{lemma}
$\dot d$ is a positive, symmetric nilpotent function with 
$$\dot d \le d$$
\end{lemma}

\begin{lemma}			

$\dot d$ satisfies the product-inequality, and $d = \dot d$ if and only if $d$ satisfies the product-inequality.


\end{lemma}

\begin{lemma}							\label{lem89}

If $d$ satisfies the product inequality then also $\ell: \hom(\ol S,\ol Q) \rightarrow \R$ the one of (\ref{eq168}) in case that $\ell(x)=d(x,\e)$. 

\end{lemma}

We proceed now analogously and thus faster as in section \ref{sec7},
and use various definitions from there, but now applied
to the homomorphism sets rather than to the object sets of $\Lambda$. 
In what follows we use 
definitions \ref{def71}, \ref{def72}, 
\ref{def74} and \ref{def53} 
applied to the homomorphism sets. 
In analogy to lemma \ref{lemma562} 
we can now repair $d$ in such a way that we get all
the desired properties we want:

\begin{lemma}						\label{lem810} 
For $\dot d$ being any positive, symmetric, nilpotent functions 
satisfying the 
product-inequality
on 
each morphism set of $\Lambda$  we have:  

$\hat {\dot d}$ satisfies the $\na$- and product-inequality. 

$\tilde {\dot d}$ satisfies the $\Delta$- and product-inequality.

$\tilde {\hat {\dot d}}$ satisfies the $\na$-, $\Delta$- and product-inequality. 

\end{lemma}

\begin{proof}

By the distribution laws of lemma 
\ref{lemma83} 
and the supposed validity of the product inequalities 
 as stated in lemma \ref{lemmulti} we have 
$$ \hat {\dot d}(QX,RX) = \inf_{QX= \na_{i=1}^n R_i,   
R X= \na_{i=1}^n R_i} 
\sum_{i=1}^n {\dot d}(Q_i,R_i)$$
$$\le  \inf_{Q= \na_{i=1}^n Q_i, 
R= \na_{i=1}^n R_i} 
\sum_{i=1}^n {\dot d}(Q_i X,R_i X) 
\le \hat {\dot d}(Q,R) \ell(X)$$
$$ \tilde {\dot d}(QX,RX) = \inf_{A_i} 
{\dot d}(Q X,A_1) +  {\dot d}(A_1,A_2) + .... + {\dot d}(A_n,RX) $$
$$\le
\inf_{A_i 
} {\dot d}(Q X,A_1 X) +  {\dot d}(A_1 X,A_2 X) + .... + {\dot d}(A_n X,RX) 
\le \tilde {\dot d}(Q,R) \ell(X)$$

Hence $\hat {\dot d}$ and $\tilde {\dot d}$ 
fulfill the product inequalities, and the rest follows from 
lemmas \ref{lemma51}, \ref{lemma54} 
and \ref{lemma562}
\end{proof}

\begin{lemma}
$d_\ell$ satisfies the product-inequality 
if and only if $\sigma_\ell$ 
satisfies it. 

\end{lemma}

\begin{proof}
By lemma \ref{lem53} we may express $\sigma_\ell$ by $d_\ell$, 
and vice versa,  and 
from this it is easily deduced. 
%
\end{proof}

We introduce now the following functions: 

\begin{definition}					\label{def820}

Define 
positive functions $\ell, \dot \ell, \tihadot \ell :
\hom (\ol S , \ol Q) \rightarrow \R$ depending on $d$ 
(recall that we do this for each morphism set $(M,d)$ of $\Lambda$ 
separately)  
by 
%
%
$$
\ell(x):= d(x,\e), \; \dot  \ell(x):= \dot d(x,\e), \; \tihadot \ell(x):= \tihadot d(x,\e)$$ 

\end{definition} 

To get finally convergence, 
it is fundamental
that the new metric candidates descend:

\begin{corollary}				\label{cor51b}

Using 
definitions \ref{def820} and \ref{def74},   we have 
$$d_{{\tihadot \ell} } \le \tihadot d, \qquad
d_{\ol {\tihadot \ell} } \le \tihadot d $$    

\end{corollary}

 \begin{proof}
 Apply corollary \ref{cor51} to $d:= \dot d$ and $\ell:=\dot \ell$. 
	\end{proof}
	 

We do now an analogous procedure as in definition \ref{def75},
but now for the homomorphism sets
instead of the object sets:

\begin{definition}					\label{def810}	

%
For all 
object sets $\ol S$ and $\ol Q$ of $\Lambda$ set the start points $d^{(1)}:= d: \hom (\ol S , \ol Q) \times \hom (\ol S , \ol Q)    \rightarrow \R$ 
and $\ell^{(1)}:= \ell:\hom (\ol S , \ol Q)   \rightarrow \R$ (def. \ref{def87}), 
and then successively  
define 
$ \tihadot {\ell^{(n)}}, \ell^{(n+1)} : 
\hom (\ol S , \ol Q)   \rightarrow \R$ 
and $d^{(n+1)}:\hom (\ol S , \ol Q)   \times \hom (\ol S , \ol Q)   \rightarrow \R$ 
($n \ge 1$) 
by 
$$ \tihadot {\ell^{(n)}}(x):= \tihadot{ d^{(n)}}(x,\emptyset)$$
$$\ell^{(n+1)} := {\ol {\tihadot {\ell^{(n)}}}}$$
$$d^{(n+1)}:= d_{\ell^{(n+1)}}$$ 

where $\dot{d^{(n)}}$ is defined as in definition \ref{def88} 
with respect to ${d^{(n)}}$ and ${\ell^{(n)}}$ (instead of $d$ and $\ell$). 
\end{definition}

The successively applied  procedure 
to a length function as described in the last definition converges again, analogously as in the theorem
\ref{cor52}, to a wished fixed point: 



\begin{theorem}[Ideal length functions on homomorphism sets]				\label{lem812}


Assume the definitions \ref{def86}, 
\ref{def87},  \ref{def810}	
and form the limits $\inft \ell:= \lim \ell^{(n)}$ and $\inft d
:= \lim d^{(n)}$ 
with respect to all morphism sets $M$ of $\Lambda$. 
 
Then for all morphism sets $M$ of $\Lambda$,
$\ell^{(n+1)}$ and 
$\inft \ell$ 
are length functions on $M$ fulfilling the product-inequality
\re{eq168}, 
$d_{\inft \ell}$ and $\sigma_{\inft \ell}$ 
are pseudometrics on $M$ fulfilling the 
$\na$- and product-inequality, 
%
and moreover and more precisely
$$
d^{(\infty)} = d_{\inft \ell}  
= \lim_{n \rightarrow \infty} d_{\ell^{(n)}} 
= \tihadot {\inft d} 
= \lim_{n \rightarrow \infty} \tihadot{d_{\ell^{(n)}}} 
= \lim_{n \rightarrow \infty} \tihadot{d^{(n)}} 
= \lim_{n \rightarrow \infty} {d^{(n)}} 
$$
$$
\qquad 
\inft \ell = \dot{\inft \ell} = \tiha{\dot{\inft \ell}} 
= \tihabar{\dot{\inft \ell}}
= \lim_{n \rightarrow \infty} \ell^{(n)}  
= \lim_{n \rightarrow \infty} \tihabar{\dot{\ell^{(n)}}}$$
 %
 %
%
$$ 
{d_{\inft \ell}}(U X,V X) \le  {d_{\inft \ell}}(U,V) {\inft \ell}(X) 
$$
$$ 
{d_{\inft \ell}}( X U, X V ) \le  {\inft \ell}(X)  {d_{\inft \ell}}(U,V) 
$$

for all composable morphisms $X,U$ and $V$ of $\Lambda$. 
\end{theorem}



\begin{proof}

We have the analogy of lemma \ref{lemma510}, that is, 
descending sequences 
$$\ell^{(n+2)}=\ol{\tihadot {\ell^{(n+1)}}}
\le \ell^{(n+1)} , 
\qquad 
d^{(n+1)} = d_{\ell^{(n+1)}} \le \tihadot{ d^{(n)} } \le  d^{(n)}
$$   
 by an analogous proof but by employing 
corollary \ref{cor51b}  
instead of \ref{cor51} ,
and so their limits $\inft d$ and $\inft \ell$ and their stated identities.  

Similar proofs 
of lemmas \ref{lem710} and \ref{lemma513} then 
yield the first 
identities.    

The last 
stated product inequalities with respect to the pair 
$(\inft d, \inft \ell)$ 
follow then simply by taking the pointwise limits 
of the valid 
product-inequalities for the pairs 
$(\tihadot {d^{(n)}}, 
{\ell^{(n)}})$ by 
definition of $\dot{d^{(n)}}$ in definition \ref{def810}	 
and lemma \ref{lem810}.  
%
\end{proof}

This is the analogy to proposition \ref{lem58}. 
It shows, for example,  when the procedure of definition \ref{def810}	 stops: 

\begin{proposition}[Ideal length functions on homomorphism sets attained]				


If $d = d_\ell$ for all morphism sets of $\Lambda$, then all $d_\ell$ 
satisfy the $\na$-, $\Delta$- and product-inequality 
$\Leftrightarrow$ 
$$d_\ell = \tihadot{ d_\ell} \quad \Leftrightarrow 
\quad 
\tihadot \ell = \ell
\quad \Leftrightarrow 
\quad 
\tihabar {\dot \ell} = \ell		\quad 
\Leftrightarrow 
\quad 
 \ell = \inft \ell$$
 
 where each item is understood to 
hold for all morphism sets of $\Lambda$. 

\end{proposition}


\begin{proof}

This is completely analogously proved as 
proposition \ref{lem58}, 
using corollary 
\ref{cor51b} and theorem \ref{lem812} instead of
 \ref{cor51} and \ref{lemma513}, respectively.
 \end{proof}

We finally may turn the obtained pseudometrics to metrics
and 
notice: 

\begin{corollary}[Quotient category]

By def. \ref{def33}, 
for each objects $\ol S ,\ol Q$ in $\Lambda$,   we
may then divide out elements in ${\hom (\ol S, \ol Q)}$ to obtain 
a metric space 
$\me {\hom (\ol S, \ol Q)}$, which is again an abelian idempotent monoid, with metric $ \me {d_{\inft  \ell}} = d_{\me {\inft \ell}}$ 
satisfying the $\na$- and product-inequality with respect to $\check{\inft \ell}$.

These quotients are   
compatible			
with composition, i.e. 
if $X_1,X_2 \in \hom(\ol S,\ol Q)$  with $[X_1]=[X_2]$ 
(class brackets) and 
$Y_1,Y_2 \in \hom(\ol Q,\ol R)$ with $[Y_1]=[Y_2]$ 
then
$$[X_1 ][Y_1]:=[X_1 \circ Y_1] = [X_2 \circ Y_2]$$ 

In this way we obtain a category $\check \Lambda$ 
out of $\Lambda$ by replacing the morphism sets $M$
by $\check M$. 

Analogous distribution formulas as in lemma \ref{lemma83} hold in $\check \Lambda$. 
\end{corollary}

\begin{proof}

The first assertions 
follow 
by definition \ref{def33} and proposition 
\ref{lem36}. 

For $d:=d_{\inft \ell}$ and $\ell:=\inft \ell$ we have 
$$ 
d(X_1  Y_1, X_2  Y_2) 
\le 
d(X_1  Y_1, X_1  Y_2 )
  + d(X_1  Y_2, X_2  Y_2)   
  $$$$
  \le 
  \ell(X_1) d(  Y_1,  Y_2 )
  + d(X_1 , X_2 ) \ell( Y_2)  
  = 0  
$$

by the $\Delta$- and product-inequality for $d$ and $\ell$ 
found out in theorem  \ref{lem812}. 
%
%
%
%
%
\end{proof}


The next lemma indicates 
how we could show that a homomorphism 
is 
uniformly continuous with respect to various derived metric candidates.

\begin{lemma}

Let $U : \ol S \rightarrow \ol Q$ be a homomorphism. 
%
 If there is an $M \ge 0$ such that 
$$d_{\ol Q} 
(U a, U b) \le M d_{\ol S} 
(a,b)$$
for all $a \le b$, then this inequality holds for all $a,b \in \ol S$.
Then also 
$$\tiha{ d_{\ol Q} 
} (Ua ,Ub)
\le M \tiha{ d_{\ol S}   
}(a,b) $$
%




Analogous assertions hold for $\sigma$ instead of $d$.

\end{lemma}

\begin{proof}

The first assertion is by the formulas of lemma \ref{lem53}, 
and the second one follows directly from definitions  
\ref{def71} and \ref{def72}.  
\end{proof}

We may apply the last theorem to get a Banach-Mazur-like distance 
between languages:

\begin{definition}[Banach-Mazur-like distance between languages]				\label{def827}

{\rm 

Let $\hom(\ol S,\ol Q)$ be equipped with a  length function $\ell$.
Then
a Banach-Mazur-like distance between $\ol S$ and $\ol Q$
(e.g. languages)
may be defined as follows:  
%
%
$$d(\ol S, \ol Q) = \log \inf \{\ell(\phi) \ell(\phi^{-1})  \in \R  |\; \phi : \ol S \rightarrow \ol Q \mbox{ an isomorphism}\}$$

}
\end{definition}

\begin{lemma}
This distance 
is a pseudometric on any  
  category of isomorphic abelian, idempotent monoids such that the hom-sets $\hom(\ol S,\ol Q)$
are provided with a length 
function $\ell$ satisfying the 
product-inequality (\ref{eq168}).
\end{lemma}

The length functions $\ell^{(\infty)}$ of theorem \ref{lem812} may be 
good candidates for the Banach-Mazur-like distance as 
they satisfy the product inequalities (\ref{eq168}) 
(indeed, 
apply the inequalities of theorem \ref{lem812} 
and note that ${d_{\inft \ell}}(x,\e)= \inft \ell(x)$).
Alternatively one may start with definitions \ref{def86} and \ref{def87} with $d=d_\ell$, 
and then using $\ell$ of lemma \ref{lem89}. 





\section{Non-increasing length functions}

In this section we work with length functions 
which may not be monotonically increasing, that is, without 
requiring inequality (\ref{eql3}). 
The length function of 
definition \ref{def29} may be replaced by the following 
possibly non-increasing bigger
one:

\begin{definition}[Length function on quotient]			\label{def91}

{\rm

Given a language $S$ with variables and a length function $L: S \rightarrow \R$ on it, set
%
$$\ell: \overline S \rightarrow \R$$
$$\ell([s]) = \inf \{ L(t) \in \R|\, t \in S, \;s \equiv t\}$$

}
\end{definition}

Now assume 
that we are given an abelian, idempotent monoid $(\ol S, \na , \e)$ with a non-increasing length function $\ell$ on it 
(i.e. without postulating \re{eql3}). 

In this section we do the same limit as in definitions \ref{def75} 
and \ref{def56},  
but instead 
of implementing the bar operator of lemma \ref{lem76} we choose other 
functions $d_\ell$ and $\sigma_\ell$ as follows, 
which do not require the monotone increasingness of $\ell$ 
and are still positive:

\begin{definition}						\label{def92}

We define for all  
$x,y \in \ol S$ and $p \ge 1$ ($p \in \R$) 
$$d_\ell(x,y) := |\ell(x \na y) - \ell(x)| \ud |\ell(x \na y) - \ell(y)|$$
%
%
$$\sigma_{p,\ell}(x,y) := \Big( |\ell(x \na y) - \ell(x)|^p + |\ell(x \na y) - \ell(y)|^p \Big)^{1/p} $$

\end{definition}

We set $\sigma_\ell := \sigma_{1,\ell}$ for short. 
We use the same notations $d$ and $\sigma$ as in 
definition \ref{def210}, as notice, there is no difference between  definitions
\ref{def92} and \ref{def210} for monotonically increasing
length functions $\ell$. 
%
%
We still have
the lower bound
$|\ell(x)-\ell(y)| \le \sigma_\ell$ 
for $\sigma_\ell$, but not so for $d_\ell$.

%
%
%

Let us be given a positive, symmetric, nilpotent function $d$ on $\ol S$. 
We are going to use definitions \ref{def71}, \ref{def72} 
and \ref{def74}.  Define $\ell$ as in definition \ref{def53}. 

\begin{lemma}							\label{lem91}
If $d$ satisfies the $\Delta$- and $\na$-inequality then 
for $\ell$ as in \re{eq167},  
$$d_\ell \le d$$
$$d_{\ol \ell} \le d $$


\end{lemma}

\begin{proof}
Yielded by a similar 
proof as for lemma \ref{lem511}.
\end{proof}

We do now a procedure as in definition \ref{def75}, 
but with the bar operator now being ignored: 


\begin{definition}					\label{def93}

Set $d^{(1)}:=  d$, and then  
for $n \ge 1$ put 
$$\ell^{(n+1)} (x) := \tiha {\ell^{(n)}}(x):= \tiha{ d^{(n)}}(x,\emptyset)$$
%
%
$$d^{(n+1)}:= d_{\ell^{(n+1)}}$$ 

\end{definition} 


\begin{corollary}

Given a positive, symmetric, idempotent function $d$ on an abelian, idempotent monoid $(\ol S,\na ,\e)$, 
define $\ell^{(n)}$ and $d^{(n)}$ 
as in definitions  \ref{def93} 
and \ref{def92}. 

We get limits $\inft \ell$ and $\inft d$ as in definition \ref{def56}, 
and the analogous assertion of 
lemma \ref{lem710}. 
That is, 
$\inft \ell$ is a non-increasing length function on $\ol S$, 
and $\inft d = d_{\inft \ell}$ is a pseudometric on $\ol S$ satisfying the 
$\na$-inequality.

\end{corollary}

\begin{proof}

By lemma \ref{lem91}, applied to $d:=\tiha {d^{(n)}}$ , 
we get
$$d^{(n+1)} = d_{\tiha {\ell^{(n)}}} \le \tiha {d^{(n)}} \le 
{d^{(n)}}$$
 
 So $d^{(n)}$ is decreasing, in particular thus also $\ell^{(n)}$
 by definition \ref{def93}, and we get a limit $\inft \ell := \lim_n \ell^{(n)}$. 
 The rest goes verbatim as in lemmas \ref{lem710} 
 and \ref{lemma513}. 
\end{proof}


\begin{proposition}				\label{lem58b}


If $d$ is of the form $d=d_\ell$ for a non-increasing length function $\ell$ then 
$d_\ell $ satisfies the $\na$- and $\Delta$-inequality 
$\Leftrightarrow$ 
$$d_\ell = \tiha{ d_\ell} \quad \Leftrightarrow 
\quad 
\tiha \ell = \ell		
%
\quad 
\Leftrightarrow 
\quad 
 \ell = \inft \ell$$

\end{proposition}

\begin{proof}

This is analogously proved as proposition \ref{lem58}. 
\end{proof}

However, it cannot be claimed that $\ell^{(\infty)}$ is increasing.

Also, contrary to 
theorem \ref{cor52}, $\delta$ of definition \ref{def62} with respect to $\inft \ell$ might not be increasing and
$\inft \sigma$ might not satisfy the $\na$- and $\Delta$-inequality,
as all the proofs given rely on the monotonically increasingness 
of $\inft \ell$. 

We could now form $\ol  {\inft \ell}$, define $d_\ell$ as above
and restart the above procedure, or the one of section \ref{sec7}.  
There seem to be endless many variants. 
 A principal problem is, 
 however this can be said, that $\ell \le \ell_2$
  might not imply $d_{\ell} \le d_{\ell_2}$,  
  so one might not be able to order the limits by $\ell$. 
  


\section{General examples}

\label{sec10}  

In this and the remaining sections, we assume that $S$ 
is the mathematical language of 
technics 
and physics, 
but also of typical mathematics. We may use quantors, set-definitions, but over all sentences need not to be completely defined
in a rigorous mathematical way.
This may also not be useful. 
In finding a new theory, it 
may be better to let the precise meaning open. 
For example, a quantum 
field theory may be 
singular and yield infinities, but after renormalization the situation turns, see the various textbooks about quantum field theory, 
for example \cite{zbMATH00822677}. 
   
   We distinguish variables in sentences which are {\em unbound}, 
   like $x$ in $x+2 =y$, and {\em bound}
   ones like $x$ in $f(x):=x+2$. 
   Here, $f$ is a function which is being {\em defined}, and in $a(y):= 2+f(y)$
   the function $f$ is being {\em called}.  
We sometimes underline variables (like so: $\ul x$) to indicate that these are main variables. 

For the length function $L$ on $S$ we choose a simple character count (i.e. $L(a_1 \ldots a_n) = n$ for $n \ge 0, a_i \in A$), without counting non-notated brackets, even if they are theoretically there (as in $x^y$), and 
set $L(:=)=1$, $L(\ud)=0$. For example, $L(x^y \ud (a+b)^2)= 8$. 
The sign $\ud$ is always 
understood to bind most weakly 
(for example, $a+b \ud c$ means $(a+b) \ud c$). 
The sign $\approx$ means approximately equal between two 
real numbers, and its more precise meaning depends 
(for example, $90 \approx 100$ for integers).
We note that for most of the discussion this  
choice of $L$ is essentially irrelevant,
and one may as well choose any other length function $L$, or
the length function $\ell = \inft \ell$ of theorem  \ref{cor52}, 
but to get specific numbers in certain examples we choose the 
character count for simplicity.

From now on, everything is to be understood to be very sloppy, 
inexact, vague and superficial!  

\begin{example}[Main and auxiliary variables]

\label{ex101}

{\rm

Main variables cannot be deleted by (\ref{eq9}). 
They are what the theory is about, for example a fundamental 
field $\ul \psi$ in physics, or the entry point $\ul {main}$ in C++.
Auxiliary variables help to simplify and shorten expressions,
typically functions which are often called.  
For example 
$$\ul x = a + b \ud a=3 \ud b=5 \qquad 
\equiv \qquad \ul x = 8 \ud a=3 \ud b=5 \qquad 
\equiv \qquad \ul x = 8$$
$$\ul f (x) = g(x)^2+g(x)^3 \ud g(x) =2x  \qquad \equiv  
\qquad 
\ul f(x) = 4x^2 + 8 x^3$$
by (\ref{eq3}) and (\ref{eq9}). 



}

\end{example}

\begin{example}[Length function $\ell$ is not directly subadditive on the level of $S$]

{\rm

Given $x,y \in S$, 
in
general we have
$$\ell([x \ud y]) \not \le  \ell([x]) + \ell([y])$$

Indeed, if $y$ is longer and contains only auxiliary variables 
and $x$ is short with main variables and uses the variables 
of $y$ (typically by function calls, that is, $x$ outsources the bulk 
of formulas to $y$ as auxilliary functions) then, by (\ref{eq9}) one has $y \equiv \e$,  and thus we got a contradiction 
 $$ 0 
 \ll \ell([x \ud y]) 
 \le  \ell([x]) + \ell([y])= \ell([x]) \approx 0$$
 
}
\end{example}

The next two examples show that even if the 
class elements 
$[x \ud y \ud \phi(x)]
\neq [x \ud y]$, their lengths 
are approximately equal.

\begin{example}[Copies are not superfluous in general]					\label{ex102}
{\rm 

 Even if $\phi(x) \bot x, \phi(x) \bot y$ for $x,y \in S, \phi \in P_S$, then - contrary to (\ref{eq8}) which says that 
 $x \ud \phi(x) \equiv x$ because $\phi(x)$ is a superfluous copy of $x$ - 
 we 
 still may have 
 \be{eq10b} 
 x \ud y \ud \phi(x) \not \equiv x \ud y
 \en
 
 Indeed, set $x:= \{\ul f := a\}$ and $y:= \{a(x):= 2x \}$, then 
 $\phi(x)= \{\ul g := b\}$ but it cannot equivalently be deleted in 
 $$  x \ud y \ud \phi(x)  \quad = \quad 
 \Big ( \ul f := a \quad \ud \quad a(x):= 2x \quad \ud \quad 
 \ul g := b \Big ) $$
 because $\ul g := b$ is a different theory than $\ul f := a$
 because $a$ is further determined, 
 whereas $b$ is not. 
 
}
\end{example}

Even if \re{eq10b} is an inequality in general, it becomes 
approximately an identity when put into $\ell$, because 
$\phi(x)$ is just a copy of $x$:

\begin{example}[Copies are cheap] 			\label{lem103}
{\rm 
If $x,y \in S, \phi \in P_S$ and $\phi(x) \bot x, \phi(x) \bot y$ then we have
$$ \ell([x \ud y \ud \phi(x)]) \approx \ell([x \ud y])$$

In other words, if we take a theory $x \ud y$, and add to it a copy $\phi(x)$ (with other variable names) of the subtheory $x$, then
the length (= cost) 
of the outcoming theory $x \ud y \ud \phi(x)$ does not much change.
This is clear for computer programs because in $\phi(x)$ we just
need to call all the routines of $x$ and need not to rewrite them
twice in $\phi(x)$.  
Indeed, going back to the concrete example \ref{ex102} we may write
$$x \ud y \ud \phi(x) \quad \equiv \quad  
\Big (  F(u,v):= 
u - v \quad\ud\quad 
F(\ul f,a)=0 \quad \ud \quad y \quad\ud \quad F(\ul g,b)=0  
\Big )$$
That is, instead of $x$ we wrote $F(u,v):= \ul u - v \ud F(\ul f,a)=0$, and instead of $\phi(x)$ $F(\ul g,b)=0$. So we have one long definition of $F$, approximately as long as $x$ itself, and two short calls 
$F(\ul f,a)=0$ and $F(\ul g,b)=0$. So we may end up with the claim:
$$ \ell([x \ud y \ud \phi(x)]) \approx 
\ell([F(u,v):= \ul u - v \ud 
F(\ul f,a)=0 \ud y]) + \ell([ F(\ul g,b)=0])
\approx \ell([x \ud y] )$$ 

In this short example it 
may not pay off, because the length of $x$ is so short, 
but if $x$ is long with few 
unbound variables, certainly. 



}
\end{example}



By definition \ref{def27}, 
the $\na$-operation  requires that $[x] \na [y] = [x \ud y]$ 
if
$x$ and $y$ have no variables in common. 
The next lemma shows when this identity approximately holds in the quotient $\me S$, see definition \ref{def33}, even if there are
common variables.

\begin{lemma}[$\na$-operation as an idealization of $\ud$-operation]							\label{lem101}
Here, $\precsim$ means approximately lower or equal. 
For all $x,y \in S$ we have  
$$ \ell([x \ud y]) \precsim  \ell([x] \na [y]) 
\quad \Rightarrow \quad d([x \ud y] , [x] \na [y]) \approx 0$$

\end{lemma}

\begin{proof}
Indeed, two-fold application of 
example \ref{lem103} yields  
$$d([x \ud y] , [x] \na [y]) \approx \ell([x \ud y \ud \phi(x) \ud \psi(y)])
- \ell([x \ud y])$$
$$\approx  \ell([x \ud y ]) - \ell([x \ud y]) = 0 $$

\end{proof}

The next two important lemmas about auxiliary and main variables are exactly proved  
in the  general framework of section \ref{sec2}.

\begin{lemma}[The more main variables the longer the sentence]
						\label{lemma106}
						
If $s \in S$ and $t \in S$ a copy of it, but where some auxiliary
variables are declared to main variables, then 
$$\ell([s]) \le \ell([t])$$

\end{lemma}

\begin{proof}

In our natural definition of $S$ there is no difference 
between main  and auxiliary variables besides the rule
\re{eq9}, which can only help to reduce the $\ell$-length
of a sentence, whence the inequality.
%
\end{proof}

The next lemma, which is an important assertion 
that holds
actually in the general framework of section \ref{sec2} as well, says that in an optimal sentence
we can change all auxiliary variables to main variables 
without changing the element in the quotient $\me S$:

\begin{lemma}	[Optimal sentences and main variables]			\label{lemma101}


If $x \in S$ 
is optimal in the sense that 
$\ell([x]) = L(x)$, 
and $x'$ is a copy of the sentence 
$x$ but where all 
unbound auxiliary variables are replaced by main variables
(in a bijective way) 
then
$$d([x],[x']) = 0 $$   
$$[x] \le [x'] \mbox{ in } \ol S \qquad 
[[x]] = [[y]] \mbox{ in } \me S $$

\end{lemma}

\begin{proof} 
%
Choose a variable transformation $\phi \in P_S$ 
such that the copy $\phi(x)$ of $x$ is disjoint to $x$
and $x'$, i.e. $\phi(x) \bot x$ 
and $\phi(x) \bot x'$.

Let $z$ be a verbatim copy of $\phi(x)$ but where all 
unbound auxiliary variables are 
declared to main variables. 
Then by lemma \ref{lemma106},  
$\ell([\phi(x) \ud x']) \le \ell([z \ud x'])$. 
By law \re{eq8}, $z$ is just another variable transformed copy of $x'$ and thus 
$z \ud x' \equiv x'$. 
Hence, 
by assumption we get
$$\ell([x] \na [x')])= \ell([\phi (x) \ud x']) \le \ell( [z \ud x']) =
\ell([x'])$$ 
$$\le L(x') = L(x) = \ell([x]) \le \ell([x] \na [x')])$$

Thus $d([x],[x']) = 0$   by definition \ref{def210}. 
%
%
%
%
\end{proof}

The following example illustrates the last lemma: 

\begin{example}[Changing auxiliary to main variables]
								\label{ex108}
{\rm 

(i)
Consider the sentence $\alpha \in S$ defined by 
$$\alpha:=\quad \Big ( \ul f := \sin \circ  u + \cos \circ  u 
\quad \ud \quad  u(x):=5 x^9+ 2x^3 +77 x^2 +1 
\Big )$$


Let $\alpha'$ be a verbatim copy of $\alpha$ but where each $u$ is replaced by $\ul u$, so 
a main variable. 
As it appears, the representations of $\alpha$ and $\alpha'$ 
are  already shortest possible and thus
\be{eq97}
\ell([\alpha]) = L(\alpha)= L(\alpha')= \ell([\alpha'])
\en
$$\ell([\alpha'] 
\na [\alpha])= 
\ell([\alpha' \;\ud \;\ul g(x):= \sin(\ul u(x)) + \cos(\ul u(x))])
=\ell([\alpha])$$
because $\ul g$  is a superfluous copy of $\ul f$ by \re{eq8}, and where we have dropped the definition of $u$ in $\alpha$ by \re{eq9} 
and taken $\ul u$ from $\alpha'$ instead.   Thus
%
$$d([\alpha],[\alpha']) = \ell([\alpha \ud \alpha']) - \ell([ \alpha])
= 
0$$

Because of \re{eq97}, 
this is in accordance with lemma \ref{lemma101}. 

(ii)
But if $u$ is a short polynomial, say $u(x) =2x$, then $\alpha$ can be simplified by saying $[\alpha]=[g(x):= \sin(2x) + \cos(2x)]$
(so  $\alpha$ is not optimal, i.e. $\ell([\alpha]) \neq L(\alpha)$), but $\alpha'$ can be not, as $\ul u$ is a main variable  
and so cannot be deleted, and so 
$$d([\alpha],[\alpha']) = \ell([\alpha \ud \alpha']) - \ell([ \alpha])
= \ell([\alpha' ]) - \ell([ \alpha])> 0$$

This example demonstrates that the assumption of lemma \ref{lemma101}  is necessary. 
}
\end{example}

\begin{example}[Approximate automaticity of \re{eq8}]

{\rm 

Relation \re{eq8}, that is, $x \ud \phi(x) \equiv x$ for $x \bot \phi(x)$, is an idealization that 
holds 
approximately automatically in our prototype language, that is, 
$\ell(x \ud \phi(x) ) \approx \ell(x)$. 
Indeed, as an illustration considering at first 
the sentence $\alpha$ of example \ref{ex108}.(i), we may write (without using \re{eq8}) 
$$\alpha \ud \phi(\alpha) \quad \equiv \quad \alpha \ud \ul g := \ul f$$
because we can shorten the long expression $\phi(\alpha)$ by 
equivalence in mathematics by just saying that we have a main variable $\ul g$ which is just a function
like $\ul f$
and so we have
$$\ell([ \alpha \ud \phi(\alpha)])= \ell([\alpha])+3$$

More generally, one drops the copy $\phi(\alpha)$ but 
rescues the main variables from $\phi(a)$ by  adding 
a list $\ul g_1 := \ul f_1 \ud  \ldots \ud \ul g_n := \ul f_n$ where 
$\ul g_i$ are the main variables appearing in $\phi(\alpha)$
and $\ul f_i$ those corresponding in $\alpha$. 

Even more, one could encode the main variables $\ul f_1, \ldots ,\ul f_n$ into one main variable $\ul f$ as a function on $\N$  and use $\ul f(1),\ldots , \ul f(n)$ instead. It would be then enough to say
$\ul g := \ul f$. 

}
\end{example}

The fact that in optimal sentences (i.e. those $x \in S$ for which $L(x) \approx \ell([x])$) 
the auxiliary variables can be turned 
to main variables 
without changing the optimality (see lemma \ref{lemma101})  can
be exploited to give an instance where
the triangle inequality is invalid:

\begin{example}[Violation of $\Delta$-inequality]			\label{ex106}

{\rm 

Assume that
$a \in  S$ has two approximately optimal 
representations $x,y\in S$ in the sense that 
$[a]=[x]=[y]$ in $\ol S$ and $\ell([a]) \approx L(x) 
\approx L(y)$. 
 Of course, 
 $$d([a],[x])=d([a],[y])= 0$$

 Since $x$ and $y$ are approximately optimal short in length, we may replace all 
 unbound auxiliary variables of $x$ and $y$, respectively, by main variables to obtain new sentences $x'$ and $y'$, respectively, 
 and still
 $$d([x],[x'])\approx 0 \qquad d([y],[y'])\approx 0$$
by lemma \ref{lemma101}.
%
 But because $x$ and $y$  
 may be very different in $S$, 
 by all the new main variables appearing now in $x'$ and $y'$,
$x'$ and $y'$ 
 contain 
 much necessary information 
 which cannot be deleted (so \re{eq9} does not apply) and it may be difficult to simplify $x' \ud \phi(y')$, 
 so we may get 
 $$d([x'],[y'] ) \gg 
 0$$ 
 But this contradicts the $\Delta$-inequality, as assuming   
 it, we got
 $$0 \approx d([x'],[a]) + d([a],[y']) \ge  d([x'],[y'] )  \gg 0 $$ 
 
 }
\end{example}

It would be helpful if one had a difference element in the sense
that if $x \le y$ for $x,y \in \ol S$ then there is a `complement' $z \in \ol S$ 
such that $x \na z=y$ ($x$ united with $z$ is the whole $y$) and 
$\zeta(x \cap z)= 0$ (the intersection of $x$ and $z$ is zero), that is, 
$$\ell(x\na z)= \ell(x) + \ell(z), $$ 
by definition \ref{def61}, or in other words, $\zeta(z) = \zeta(y \backslash x)$. 
But it does not exist in general:

\begin{example}[Difference Element does not exist]			\label{ex108}

{\rm

Consider two formulas $\alpha, \gamma \in S$ as follows: 
$$\alpha:= \quad  \Big ( \ul b(x):= x^2+1 
\quad \ud \quad \ul y := b^2 +b \Big )$$
%
%
$$ 
\gamma:= \quad \Big ( \ul z(x) := x^4+3 x^2+2
 \Big ) $$

We assume that $\alpha$ and $\gamma$, as it appears, are already in their shortest form, so,  
observing also that $\ul y = \ul z$ are the same polynomials, 
and to this applying \re{eq8} and then \re{eq3},
we get  
$$[\alpha] \na   
[\gamma] = 
[\ul b(x):= x^2+1  
 \;\ud \;\ul y (v) = v^4+3 v^2+2 \;\ud \;  \ul z(t) := t^4+3 t^2+2
 ]=[\alpha]$$
$$\ell([\alpha]) = L(\alpha)=15,  \quad \ell([\gamma]) = L(\gamma)=13, \quad [\gamma] \le [\alpha], \quad 
\ell([\alpha] 
\na [\gamma])  = \ell([\alpha]) $$
%
%
%
$$\zeta([\alpha] \backslash [\gamma])=
\ell([\alpha] \na  
[\gamma]) - \ell([\gamma]) = d([\alpha],[\gamma])= 2$$


The last line shows that a difference element $[\alpha] \backslash [\gamma] $ 
in $\ol S$ does not exist, as it should have only length $2$,
which appears impossible in our language $S$.

}
\end{example}

\section{Examples from Physics}

In  this section we collect various examples from physics,
and indicate  that the distance 
between an existing and
a new improved theory is often very low, as well as the lengths
of the theories itself. 
A very good instance is already the transition form the Klein-Gordon equation 
to the Dirac equation discussed in example \ref{ex213} and
the paragraph before it. 
Another example would be Maxwell's equations 
extended 
from Ampere's and Faraday's equations.

\begin{example}[Quantum mechanics]			\label{ex70}

{\rm 

Classical mechanics can be 
described as follows.
In these equations, $\dot p_i$ means derivative of $p_i$ after time,
and $\dot q_i$ a free independent variable 
$$p_i =  \frac{\partial\call}{ \partial \dot q_i }  (q,\dot q,t)
\qquad \mbox{(compute $\dot q_i = \dot q_i(q,p,t)$ from these equations)}$$
$$H (q,p,t)= \sum_{i=1}^n \dot q_i  p_i - 
\call   \qquad \mbox{(Hamilton function)}$$
%
$$\dot p_i = - \frac {\partial H}{\partial q_i}
\quad \ud  \quad \dot q_i = \frac {\partial H}{\partial p_i}
\qquad \mbox{(Hamilton equations)}$$


If the Lagrangian $\call = T -U$ (kinetic 
energy 
$-$ potential energy) then
$H = E$ (energy). The Schrödinger equations \cite{zbMATH02589457, zbMATH02589458} read 
\be{se1}
\hat p_i = - i h \frac{\partial}{\partial q_i}  
\en
\be{se2}
i h\frac {\partial \psi}{\partial t}  (q, \hat p,t) = H(q,\hat p,t) \psi(q,t)
\en

We abbreviate quantum mechanics by $QM$, Schrödinger equations \re{se1}-\re{se2} by $SE$, Hamilton function by $H$, Hamilton equations by $HE$ and classical mechanics by $CM$. 
The distance between $QM= SE \ud H$ and $CM = HE \ud H$ is 
  %
$$d([QM],[CM]) = \ell([QM] \na [CM]) - \ell([CM])
=$$ 
$$\ell([SE \ud H \ud \phi(HE) \ud \phi(H)])   - \ell([CM])$$

for a suitable $\phi \in P_S$ by def. \ref{def27}. 
Since $H$ and $\phi(H)$ are copies of each other, we may take
the Hamilton function $H$ instead of its copy $\phi(H)$ in $\phi(HE)$
and so choose a variable transformation $\psi \in P_S$ accordingly 
and then dropping the superfluous term
$\phi(H)$ which only consists of auxiliary variables and which is not needed anymore by \re{eq9}. So we get
$$=\ell([SE \ud H \ud \psi(HE)])   - \ell([HE \ud H])$$
\be{eq166}
\approx \ell([SE]) +\ell([H] + \ell([\psi(HE)])   - \ell([HE])
- \ell([H]) 
\en
$$=\ell([SE]) \approx 100 $$ 
a low number, for all Hamilton functions $H$. 

In \re{eq166} we have assumed that there is not much optimization 
possible between the involved parts, in particular in relation to a
longer Hamilton function $H$.

 

The $\na$-inequality is somehow the idealization of an 
estimate as above. Compare 
it
with (cf. also lemma \ref{lem101}) the fast approximate estimate
by definition \ref{def41}, provided the $\na$-inequality holds		 for $d$,   
$$d([SE] \na [H],[HE] \na [H]) \le d([SE],[HE])$$ 
}
\end{example}

\begin{example}[General relativity]

{\rm 
In generalizing the field equation of Newtonian mechanics 
to relativity, 
Einstein noted that because of the equation $E=m c^2$,
a perpetuum mobile could be formed by sending a light beam upwards from a big mass like the earth and so 
lifting mass upwards in a gravitational field without using energy.

To avoid violation of the theorem of conservation of energy, 
he suggested that spacetime is deformed under the influence 
of mass,  
 so that the upwards-traveling lightbeam could shift to red and so lose energy. 
To this end, mass is substituted by energy, 
which is turned to a tensor to obtain the energy 
momentum tensor $T$, 
which 
interacts with the metric $g$ of the 4-dim spacetime by the formula \cite{zbMATH02618644} 
$$R_{\lambda \mu} - \frac{R}{2}  
g_{\lambda \mu} = - \frac{8  G  \pi} {c^{4}}  T_{\lambda \mu} $$

Once again, 
the distance between general relativity GR and Newtonian field equations NF appears very low, i.e. 
$$d([GR],[NF]) \approx 100$$
maybe the shortest possible of a solution extending SR, NF and
solving the perpetuum mobile problem. 
Simply already because the theory of GR is short, i.e.
$\ell([GR])$ is low (cf. lemma \ref{lem34}). 

Since Maxwell's equations are Lorentz invariant, and $E=m c^2$ 
is deduced from Maxwell's equations, GR might already follow 
from the latter by optimization, 
because if nature were not Lorentz 
invariant, but the Maxwell's equations hold, the whole alternative 
theory 
might be longer. 
}
\end{example}




\begin{problem}					\label{prob1}

{\rm 

Given any theory $x$ with singularities, 
indicate a theory $y$ without singularities
(i.e. without mathematical pathological inconsistencies 
and divergences) 
and such that $y$ works approximately as $x$ in the range
where $x$ works and
$$d(x,y)= \mbox{low} $$		

Even the low distance alone might increase the likelihood 
that $y$ generalizes $x$ properly. 
 


 
}
\end{problem}		

\section{Examples of Computing}

For computer languages we use 
the definitions of 
example \ref{ex22}. 


\begin{lemma}					\label{lem121}


Assume that $\phi: A \rightarrow B$ is an 
isomorphism between two 
computer languages $A,B$ (exceptionally with only  finitely
many main variables available). 

  
Then there is a constant $M \ge 0$ such that for all $s,t \in A$
$$|d_B(\phi(s),\phi(t)) - d_A(s,t) | \le M$$

Hence
$$\lim_{d_A(s,t) \rightarrow \infty } \frac{d(\phi(s),\phi(t))}{d(s,t)}= 1$$
 
 \end{lemma}
 
 \begin{proof}
Indeed,    
write a compiler or interpreter $i \in B$ such that
%
$$\phi(s) \quad \equiv  \quad i \ud \tilde s \quad = \quad 
\Big ( i \quad \ud \quad a:= s' \quad  \ud 
\quad b \Big ) $$

 does the same on the computer $B$ as $\phi(s)$ does 
 by compiling or interpreting the code $s \in A$ on the computer 
  $B$. We assume 
  since $\phi$ is an isomorphism this is possible, and that $\phi$ is given on the formal language level as in definition \ref{def81} in this proof. Here $s'$ means an optimal short program $s$ (i.e. $\ell_{A}([s])= L_A(s)$) as a pure text string, which we assign the variable $a$
  above. Further,  $b$ means some minor code 
 containing 
 all main variables of $\phi([s])$ and calling $i$ to interpret or compile $a$. Moreover, 
  $\tilde s := (a:= s' \ud b)$  
 is just 
 an abbreviation. 

We also assume that 
$$L_A(s) = L(s')$$ 
gives only the length 
of the 
text string $s'$ and similarly so for $L_B(t)$. 
Also, 
as we have only finitely many main variables at all (say one 
entry point of a program is the main variable), 
we assume that there is a fixed constant $K$ 
such that 
$$L_B(b) < K$$


for all $s \in S$. We have 
$$  (i \ud \tilde s ) \na ( i \ud \tilde t ) \equiv  i \ud \tilde s \ud j \ud \tilde t 
\equiv  i \ud \tilde s \ud \tilde t 
\equiv i \ud \widetilde{s \ud   t} \equiv i \ud \tilde x$$

for any $x= s \ud t$. 
Here, the second equality is because we do not need a copy $j$ of the compiler $i$, the third one is because instead of interpreting $s'$ and $t'$ separately, we may interpret $(s \ud t)'$ at once  to get the same code.  

Hence, setting $M:=   L_B(i) + 3 + K $
we get 
$$\ell_B([\phi(s)])
= \ell_B([i \ud \tilde s]) \le L_B(i) +  L_B(\tilde s)
\le  L_B(i) + L(s') + 3 + L_B(b) 
$$$$
\le M + L(s')
= M+ L_A(s) = \ell_A([s]) + M$$


Analogously we get
$$\ell_A([\phi^{-1}(t)]) \le \ell_B([t]) + M$$
so that we end up with 
$$|\ell_B([\phi(s)]) - \ell_A([s])| \le M$$


Thus 
$$|\sigma([i \ud \tilde s] ,[ i \ud \tilde t]) - \sigma([s],[t])|$$
$$\approx |2 \ell_B( [i \ud \tilde x]) - \ell_B([i \ud \tilde s]) - \ell_B([i \ud \tilde t])  
 - 2 \ell_A([x]) + \ell_A([s]) + \ell_A([t])|$$
 $$\le 4 M$$
 
 \end{proof}
 

We are going to define a Reduced Instruction Set Central
Processing Unit (RISC CPU) in our ordinary mathematical 
language of physics and technics by describing its function of 
 memory, 
accumulator and program pointer in spacetime:

\begin{example}

{\rm

A simple computer can be realized as formulas like so: 

$$x_{t+1,n} = a_t [x_{t , p_t} = S \ad x_{t , p_t+1}=n] + x_{t,n} 
(1-[x_{t , p_t} = S \ad x_{t , p_t+1}=n])	$$

\begin{eqnarray*}
 a_{t+1} &=& x_{t,p_t+1} [x_{t,p_t} = L]
+( a_t + x_{t,p_t+1} ) [x_{t,p_t} = A]  \\
&& +( a_t * x_{t,p_t+1} ) [x_{t,p_t} = M] 
+ a_t [x_{t,p_t} \notin \{L,A,M\}]
\end{eqnarray*}

$$p_{t+1} = x_{t,p_t+1} [x_{t,p_t} =B \ad a_t \ge 0] +
(p_t +2) (1-[x_{t,p_t} =B \ad a_t \ge 0])$$

where $[A]=1 \in\R$ if an assertion $A$ is true, and $[A]=0 \in \R$ otherwise.
%

Here, $x:\Z \times \Z \rightarrow \R$ stands for memory, $a: \Z \rightarrow \R$ for an accumulator, and $p: \Z \rightarrow \R$ for program pointer. Everything is time indexed $t \in \Z$. 
Then $x_{t,n} \in\R$ is the value of memory at place $n \in \Z$ at time $t \in \N$. Similarly, $a_t,p_t \in \R$ is the value of accumulator, program pointer, respectively,  at time $t\in \N$  . 
We have the opcodes $S,L,A,M,B$ (some constant real values) which stand for store accumulator, load accumulator, add accumulator, multiply accumulator, and branch. 
 One command has the format opcode and value in the next memory. 
For example
$$ L  1000 \quad B 2000$$
mean, load accu with the value at memory address 1000, or,
branch to 2000 if accu $\ge 0$. 
 
If a sentence $s \in S$ is longer, it may pay off to implement
the above computer and implement the formula $s$ 
via  a program. 
In other words, 
to add the above computer $c$ as auxiliary variables, i.e. 
$s \equiv s \ud c$ and  simplify $s \ud c$ 
by writing a program on $c$ which substitutes (part of) $s$. 
In this way, long sentences are in concurrence with computing, 
which emerges out of nowhere simply because of the pressure
of optimization.  

The above formulas are similarly to 
differential equations in rough shape.  
One could analogously allow functions than real values, 
that is, $x: \Z \times \Z \rightarrow C^\infty(\R)$ and add a differential command for the accumulator, say, $a_{t+1} = d a_t/d y$. 

}
\end{example}



We may 
enrich any language without variables 
with variables to get a macro language:

\begin{example}[Macro language]		\label{ex160}

{\rm

Given a language $S$, with or without variables, we may form a macro language $M$ by adding countably infinitely many auxiliary and 
main variables,
function definitions (macros) which accept text strings as input 
and give a text string as output, function calls, concatenation of
text (written $.$), 
and a $\ud$ sign. 
Sentences of 
$M$ are sentences of $S$ enriched with macros. 

For example, both $\ul x$ are the same: 
$$ \ul x:=  \Big ( f(`aa','b') \ud f(x,y):=  x . x . y . `teee' .y   \Big ) 
\quad = \quad \ul x:='aaaabteeeb' $$ 

If two sentences $s,t$ in $S$ are similar by much text overlap, the distance $d(s,t)$ becomes low in the macro language $T$ by using macros. 
 
}
\end{example}

\begin{example}[Further examples]

{\rm 

(i)
In the usual mathematical framework we could also 
allow Boolean 
algebra expressions like 
$$A:= \quad \ul a = x_1 \ol x_2 + x_9 (\ol x_7 +  x_6) + ....$$
and compare
$d(A,B)$.  
More generally we may compare any two functions $f,g$ 
which allow definitions by symbolic expressions by considering
$d(f,g)$. 


(ii)
Further we could work with text strings 
and concatenation in our general mathematical framework
and 
in this way get a stronger macro language 
than in example  \ref{ex160} as we  had natural numbers, functions or pointers 
to functions as arguments in functions calls etc.

(iii)
We may consider set definitions and in this way compute 
the distance between sets, or algebraic structured sets like rings etc.

(iv) Mathematical proofs may be compared by writing down proofs
$a $, regarding them as text strings (or better as text strings in the framework of the mathematical logic of deduction), and compare these text strings within the macro language as  
in example \ref{ex160} or in item (ii).

 
}
\end{example}

\bibliographystyle{plain}
\bibliography{refer}

\end{document}